\documentclass[11pt]{article}
\def\final{1}
\def\full{1}
\def\extras{0}
\usepackage[usenames,dvipsnames]{color}
\usepackage{xspace,epsfig,wrapfig,pifont,fullpage}
\usepackage{url}
\usepackage{graphicx}
\usepackage{paralist}
\usepackage{amsmath,amssymb,amsthm}
\usepackage{latexsym}
\usepackage{amsfonts}
\usepackage[numbers]{natbib}
\usepackage[ruled,vlined,linesnumbered]{algorithm2e}
\usepackage{hyperref}

\ifnum\final=0
\newcommand{\mynote}[1]{\marginpar{\tiny\sf #1}}
\else
\newcommand{\mynote}[1]{}
\fi
\newcommand{\anote}[1]{\mynote{\color{BrickRed} Adam: {#1}}}

\newcommand{\snote}[1]{\mynote{Sofya: {#1}}}

\newcommand{\mypar}[1]{\paragraph{#1}}
\newcommand{\drop}[1]{}

\def \R {\mathbb{R}}

\newcommand{\E}{\mathbb{E}}

\newcommand{\cala}{\mathcal A}

\newtheorem{theorem}{Theorem}[section]
\newtheorem{lemma}[theorem]{Lemma}
\newtheorem{prop}[theorem]{Proposition}

\newtheorem{corollary}[theorem]{Corollary}

\newtheorem{assumption}[theorem]{Assumption}

\newtheorem{definition}[theorem]{Definition}
\newtheorem{defn}[theorem]{Definition}

\newcommand{\ie}{{\it i.e.,\ }}

\newcommand{\eps}{\epsilon}

\def\E{\mathop{\mathbb{E}}\displaylimits}

\def\Lap{\mathop{\rm{Lap}}\nolimits}

\newcommand{\re}{\mathbb{R}}

\newcommand{\defeq}{\overset{\text{def}}{=}}
\newcommand{\argmin}{\operatorname*{argmin}}

\newcommand{\ip}[1]{{\left \langle {#1} \right\rangle}}
\newcommand{\A}{\mathcal{A}}
\newcommand{\G}{\mathcal{G}}
\newcommand{\Gn}{\mathcal{G}_n}
\newcommand{\Gd}{\mathcal{G}^{\thresh}}
\newcommand{\Gnd}{\mathcal{G}_n^\thresh}
\newcommand{\dwire}{d_{\text{\rm node}}}
\newcommand{\Av}{\bar{d}} %{\text{\rm Av}}
\newcommand{\cum}{P} %{\text{\rm Av}} %cumulative degree distribution
\newcommand{\dedge}{d_{\text{\rm edge}}}

\newcommand{\paren}[1]{\left({#1} \right)}

\newcommand{\onesmat}[1]{{\mathbf{1}_{{#1}}}}
\newcommand{\outfl}[1]{{#1}_{s\bullet}}
\newcommand{\infl}[1]{{#1}_{\bullet t}}
\newcommand{\inoutfl}[1]{{#1}_{s\bullet,\bullet t}}
\newcommand{\Dvec}{\vec{\thresh}_{2n}}

\newcommand{\thresh}{{D}}

\newcommand{\vnew}{v^{\text{new}}}
\newcommand{\FG}{{\text{\sf FG}}}
\newcommand{\deglist}{\text{deg-list}}

\newcommand{\dlext}{\hat f_\thresh}
\newcommand{\dlexthat}{\hat f_{\hat \thresh}}

\begin{document}
\title{Efficient Lipschitz Extensions \\for High-Dimensional Graph
  Statistics \\ and Node Private Degree Distributions}
%\title{Efficient Lipschitz Extensions \\for High-Dimensional Graph
%  Statistics \\ with Applications to Data Privacy}

\author{Sofya Raskhodnikova\thanks{Computer Science and Engineering
    Department, Pennsylvania State
    University. \url{{asmith,sofya}@cse.psu.edu}. Supported by NSF
    awards  CDI-0941553 and IIS-1447700 and a Google Faculty Award. Part of this work was done
    while visiting Boston University's Hariri Institute for Computation.} \and Adam Smith\footnotemark[1]}

\thispagestyle{empty}
\maketitle

\begin{abstract}
Lipschitz extensions were recently proposed as a tool for designing node differentially private algorithms. However, efficiently computable Lipschitz extensions were known only for 1-dimensional functions (that is, functions that output a single real value). In this paper, we study efficiently computable Lipschitz extensions for multi-dimensional (that is, vector-valued) functions on graphs. We show that, unlike for 1-dimensional functions, Lipschitz extensions of higher-dimensional functions on graphs do not always exist, even with a non-unit stretch. We design Lipschitz extensions with small stretch for the sorted degree list and for the degree distribution of a graph. Crucially, our extensions are efficiently computable.

We also develop new tools for employing Lipschitz extensions in the design of differentially private algorithms.  Specifically, we generalize the exponential mechanism, a widely used tool in data privacy. The exponential mechanism is given a collection of score functions that map datasets to real values. It attempts to return the name of the function with nearly minimum value on the data set. Our generalized exponential mechanism provides better accuracy when the sensitivity of an optimal score function is much smaller than the maximum sensitivity of score functions.

We use our Lipschitz extension and the generalized exponential mechanism to design a node-differentially private algorithm for releasing an approximation to the degree distribution of a graph. Our algorithm is much more accurate than algorithms from previous work.

%~\cite{BBDS13,KNRS13}.
%
%
%  We provide new algorithms for estimating the degree distribution of
%  a graph while ensuring node-level differential privacy. The main
%  technical tool we introduce is a Lipschitz extension $\hat
%  f_\thresh$ of the degree distribution, with the following
%  properties: given a graph $G$
%
\end{abstract}

\ifnum\full=1
\newpage
\thispagestyle{empty}
\tableofcontents
\fi

\newpage
\pagenumbering{arabic}
\section{Introduction}\label{sec:intro}
The area of {\em differential privacy} studies how to output global information contained in a database while protecting privacy of individuals whose information it contains. Typically, the datasets considered are {\em tabular databases}, containing one row of information per person. While the area came a long way in the last decade in terms of the richness of information that can be released with differential privacy for tabular databases, we are lagging behind in our understanding of {\em graph} datasets that also contain relationships between various participants. Such datasets are used, for example, to capture relationships between people in a social network, communication patterns, and romantic relationships.

There are two natural variants of differential privacy that are suited
for graph datasets: {\em edge} differential privacy and {\em node}
differential privacy. Intuitively, the former protects relationships
among individuals, while the latter protects each individual, together
with all his/her relationships.
Edge privacy is a weaker notion and has been studied more extensively, with
algorithms now known for the release of subgraph counts and related scalar-valued functions
\citep{NRS07,RHMS09,KRSY11,MirW12,LuM14exponential,KarwaSK14}, the
degree distribution
\citep{HayLMJ09,HayRMS10,KarwaS12,LinKifer13,KarwaS15}, cut densities
\citep{GRU12,BlockiBDS12jl} and the
parameters of generative graph models
\citep{MirW12,KarwaSK14,LuM14exponential,KarwaS15,XiaoCT14}.
% Node differential privacy is more in
% the spirit of protecting each individual's data, but is much harder to
% achieve because it guards against larger changes in the input.
% Node privacy, however, provides
% more meaningful semantics when nodes correspond to individuals: it
% ensures that \emph{no matter what an analyst observing the released
% information knows ahead of time}, she learns the
% same things about an individual Alice regardless of whether
% Alice's data are used or not (see \cite{KS08} for a formal
% statement). In particular, no assumptions are needed on the way the
% individuals' data are generated (they need not even be
% independent).\snote{Why is this in the discussion of node vs. edge differential privacy? Most of it applies to both.} Unfortunately, differential privacy's stringency
% makes the design of accurate, \emph{node}-private algorithms challenging.\snote{Why is this sentence here? What was the main idea in this paragraph? What is in the next one?}
%
Node differential privacy is a much stronger privacy guarantee,  but is much harder to
attain because it guards against larger changes in the input.
Until recently, there were no known
differentially private algorithms that gave accurate answers on
sparse graphs, even for extremely simple statistics. %The two  known
%techniques were proposed %in a set of independent works that came out
In 2013, \citet{BBDS13,KNRS13,ChenZ13} proposed two new techniques for node private algorithms:
\begin{inparaenum}\renewcommand{\labelenumi}{(\roman{enumi})}
\item using
  projections whose smooth sensitivity could be bounded (combined with
  mechanisms that add noise tailored to the smooth sensitivity~\citep{NRS07}), and
\item using {\em
    Lipschitz extensions} (combined with the standard Laplace mechanism).
\end{inparaenum}
The latter technique yielded much more accurate
algorithms than the former. In particular, it was used to obtain
accurate node differentially private
algorithms for computing subgraph counts and related statistics.

However, efficiently computable Lipschitz extensions were known only
for 1-dimensional functions (that is, functions that output a single
real value). In this paper, we study efficiently computable Lipschitz extensions
for multi-dimensional (that is, vector-valued) functions. We show that, unlike for 1-dimensional functions, Lipschitz extensions of higher-dimensional functions do not always exist, even with a non-unit stretch. We design Lipschitz extensions with small stretch for the sorted degree list and for the degree distribution of a graph. Our extensions can be computed in polynomial time.

We also develop new tools for employing Lipschitz extensions in the
design of differentially private algorithms.  Specifically, we
generalize the exponential mechanism of \citet{MT07}, a widely used
tool in data privacy.
% The exponential mechanism is given a collection
% of score functions that map datasets to real values. It attempts to
% return the name of the function with nearly minimum value on the data
% set.
Our generalized mechanism provides better accuracy
when the sensitivity of an optimal score function is much smaller than
the maximum sensitivity of score functions.

% We combine these tools\snote{Which tools? Tools for employing Lipschitz extensions? What are we combining them with?} to provide a differentially
% private algorithm for releasing an approximation to the degree distribution of a graph
% that is much more accurate than those from previous
% work~\cite{BBDS13,KNRS13}.

We use our Lipschitz extension and the generalized exponential mechanism to design a node differentially
private algorithm for releasing an approximation to the degree distribution of a graph. Our algorithm is much more accurate than those from previous
work~\cite{BBDS13,KNRS13}.

\paragraph{Lipschitz extensions.} Lipschitz extensions are basic
mathematical objects studied in functional analysis.

\begin{definition}[Lipschitz constant]\label{def:lip-const}
  Let $f: X\to Y$ be a function from a domain $X$ to a range $Y$ with
  associated distance measures $d_X$ and $d_Y$.  Function $f$ has
  Lipschitz constant $c$ (equivalently, is $c$-Lipschitz) if
  $d_Y(f(x),f(x'))\leq c\cdot d_X(x,x')$ for all $x,x'\in X$.
\end{definition}

\begin{definition}[Lipschitz extension]\label{def:lip-ext}
  Consider a domain $X$ and a range $Y$ with associated distance
  measures $d_X$ and $d_Y$, and let $X'\subset X$. Fix constants $c>0$
  and $s\geq 1$.  Given a $c$-Lipschitz function $f':X'\to Y$, a
  function $f: X\to Y$ is a {\em Lipschitz extension of $f'$ from $X'$
    to $X$ with stretch $s$} if
\begin{enumerate}
\item $f$ is an {\em extension of $f'$}, that is, $f(x)=f'(x)$ on all $x\in X'$ and
\item $f$ is $s\cdot c$-Lipschitz.
\end{enumerate}
If $s=1$, then we call $f$ a {\em Lipschitz extension of $f'$ from $X'$ to $X$} (omitting the stretch).
\end{definition}

Functional analysts have devoted considerable attention to
determining, for given metric spaces $X,X'$ and $Y$, whether
Lipschitz extensions with stretch 1 exist for all functions $f:X\to
Y$. In contrast to this paper, the focus is mostly on continuous
function spaces.

Lipschitz extensions of real-valued 1-dimensional functions with stretch 1 always exist~\cite{McShane34}. We show that it is not true, in general, for multi-dimensional functions on graphs, even with non-unit stretch. The technical core of this paper
is the construction of an efficiently computable extension of the
degree distribution, a
high-dimensional function on graphs, with small stretch.

\paragraph{Metrics on Graphs.}
Let $\G$ denote the set of all finite labeled, unweighted undirected graphs.
When the input data set is a graph in $\G$, there are two natural notions of
``neighbor'' (or adjacency).
Two graphs $G$ and $G'$ are {\em edge neighbors} if they differ in one
edge. Two graphs $G$ and $G'$ are {\em node neighbors} if one can be
obtained from the other by removing one node and its adjacent edges.
These two notions of neighbor induce two metrics on $\G$, node
distance ($\dwire$) and edge distance ($\dedge$).

\paragraph{Why are Lipschitz Extensions Useful for Privacy?}
A randomized algorithm $\A$ is
\emph{node differentially private} if, for any two datasets that are ``neighbors''
in an appropriate sense, the \emph{distributions} on the algorithms
outputs are close in a multiplicative sense.
Notions of stability and sensitivity play a key role in the design of
differentially private algorithms. Differential privacy itself can be
seen as a stability requirement, since the algorithm must map
neighboring graphs to nearby distributions on outputs.

The two basic building blocks for designing differentially private algorithms,
the Laplace and exponential mechanisms, rely on the \emph{global
  sensitivity} of a function $f$, which is the Lipschitz constant of
$f$ viewed as a map from data sets (e.g., $\G$ equipped with $\dwire$)
to $\ell_1^p$ (\ie $\re^p$ equipped with $\ell_1$).
The \emph{Laplace mechanism} \citep{DMNS06} shows that one can satisfy differential
privacy by releasing $f(G)$ with additive noise proportional to the
node global sensitivity in each coordinate.

The difficulty with employing the Laplace mechanism directly is that many useful functions on graphs are highly
sensitive to the insertion or removal of a well-connected vertex. For
example, the number of connected components of a graph may go from $n$
to 1 with the insertion of a single vertex. The degree distribution of
a graph can also change drastically, shifting up by 1 in every coordinate (as
one vertex can increase the degree of all other vertices). This
difficulty generally remains even if we shift from global sensitivity
to more local notions (as in \citep{NRS07}) (roughly, interesting
graphs such as those with low average degree are ``near'' other graphs
with a vastly different value for the function).

One can get around this by focusing on a
``nice'' or ``typical'' subset of the space $\G$ where the function $f$ has low
global sensitivity\citep{BBDS13,KNRS13,ChenZ13}. For example, let $\Gd$ be the set
\emph{$\thresh$-bounded} graphs, that is, graphs of maximum degree at most
$\thresh$. Many functions have bounded sensitivity (Lipschitz
constant) on $\Gd$. The
number of triangles in a graph, for instance, changes by at most
$\binom \thresh 2$ among node-neighboring graphs of degree at most
$\thresh$, and the degree list changes by at most $2\thresh$ in
$\ell_1$.

Given a function $f$ that has low Lipschitz constant on ``nice''
graphs, if we find an efficiently computable Lipschitz extension $\hat
f$ that is defined on all of $\G$, then we can use the Laplace
mechanism to release $\hat f(G)$ with relatively small additive
noise. The lower the stretch of the extension, the lower the overall
noise. The result will be accurate when the input indeed falls into,
or near, the class of ``nice'' graphs.  Interestingly, the class of
``nice'' graphs need not contain the input for the answer to be
accurate---in our main application, we use $\Gd$ as the set of
``nice'' graphs, but $\thresh$ is set much lower than the
actual maximum degree of the input.

% Instead; focus on a class of ``likely'' or ``nice'' inputs (this
% paper: low-degree graphs). Look at functions that are smooth on those
% graphs. Extend to $\hat f$ defined on all graphs. Add noise to $\hat
% f$. Stretch $c$ is the extra amount you pay in noise.

\paragraph{Existence and efficiency of Lipschitz extensions.}
Motivated by this methodology, we ask: when do Lipschitz extensions exist, and when
do they admit efficient algorithms? The existence question has drawn
interest from functional analysis and combinatorics for nearly a
century~\citep{McShane34,Kirszbraun1934,Wells75embeddingsbook,MarcusP84,JohnsonLS86,Matousek90,Ball92,BenyaminiL98book,LangPS00,Naor01phase,LeeN05,MakarychevM10}; see
\citet{LeeN05} for an overview.
For any real-valued function $f:\Gd\to\re$, there exists an
extension $\hat f:\G\to\re$ whose node sensitivity is the same as that
of $f$.
\citet{KNRS13,ChenZ13} constructed \emph{polynomial-time} computable Lipschitz extensions from $\Gd$ to
$\G$ of several real-valued functions on graphs. The techniques
in~\cite{BBDS13,KNRS13,ChenZ13} apply to functions that count structures in a
graph, possibly with weights (for example, the number of edges in a
graph, the number of triangles in a simple graph; in a graph where
vertices and edges have attributes, one could count edges that link
nodes labeled by different genders in a social network, or triangles
involving vertices labeled with different scientific fields in a
collaboration graph).

Prior work on constructions of higher-dimensional
extensions focused on extending functions on a metric space $X$, where $X$ is given
explicitly as
input (say, as a distance matrix)~\cite{LeeN05,MakarychevM10}. Such constructions can, at best,
run in time polynomial in the size of $X$. The size of
$\Gd$ is infinite, and even restricting to graphs on at most $n$
vertices leaves a set that is exponentially large in $n$. Moreover,
generic constructions have stretch at least polynomial in the log of
the metric's cardinality, at least $\sqrt{n}$ in our case.

\subsection{Our Contributions} In this paper, we demonstrate that efficient and nontrivial
constructions of Lipschitz extensions for high-dimensional graph
summaries are possible. We also develop new machinery for using these
extensions in the context of differentially private algorithms.

\paragraph{Lipschitz Extension of the Degree List%
\ifnum\full=1 \ (Section~\ref{sec:degree-list-extension})\fi.} Our main technical
contribution is a polynomial-time, constant-stretch Lipschitz
extension of the \emph{sorted degree list}, viewed as a function from $\Gd$
to $\ell_1^*$, to all of $\G$. Here $\ell_1^*$ denotes the $\ell_1$ metric on the space
of finite-length real sequences, where sequences of different length
are padded with 0's to compute distance.

Given an arbitrary graph $G$, our function $\hat f_\thresh(G)$
outputs a nonincreasing real sequence of length $|V_G|$. If the maximum degree of $G$ is $\thresh$ or
less, the output is the sorted list of degrees in $G$.
 The output can be thought of as a list of ``fractional degrees'',
where ``fractional edges'' are real weights in $[0,1]$ and the ``fractional
degree'' of a vertex is the sum of the weights of its adjacent edges.  The weights
are selected by minimizing a quadratic function over the polytope of
$s$-$t$ flows in a directed graph closely related to $G$. Previous
work \citep{KNRS13} had shown that the \emph{value} of the maximum
flow in the graph has low sensitivity; by introducing the quadratic
penalty, we give a way to select an optimal flow that changes slowly
as the graph itself changes.
Introducing a strongly convex penalty (or \emph{regularizer}) to make
the solution of an optimization problem stable to changes in the
loss function
is common in machine learning. In our setting, however, it is the \emph{constraints} of the
convex program that change with data, and not the loss function.

\begin{theorem}\label{thm:lip-ext-degree-list-intro}
There is a Lipschitz extension of $\deglist$, viewed as a function
taking values in $\ell_1^*$, from $\Gd$ to $\G$ with stretch 3/2 that can be computed in polynomial time.
\end{theorem}

The sorted degree list has $\ell_1$ sensitivity $\thresh$ on
$\Gd$. The extension $\hat f_\thresh(G)$ has $\ell_1$ sensitivity at
most $3\thresh$ (the stretch is thus at most $3/2$). Previous results
on Lipschitz extensions only imply the existence of an extension with
stretch at least $n$; see
\ifnum\full=1 Section~\ref{sec:general-lip-ext}
\else the full version
\fi for
discussion of the general results.

We use our Lipschitz extension of the sorted degree list to get a Lipschitz extension of the
degree distribution (a list of counts of nodes of each
degree) and the degree CDF (a list of counts of nodes of at
least each given degree). These functions condense the information to a
$\thresh$-dimensional vector (regardless of the size of the graph),
making it easier to release with node-differential privacy.

\paragraph{Generalized Exponential Mechanism for Scores of Varying
  Sensitivity%
\ifnum\full=1\  (Section~\ref{sec:newexpmech})\fi.}
One of the difficulties that arises in using Lipschitz extensions for
differentially private algorithms is selecting a good class of inputs
from which to extend. For example, to apply our degree distribution
extension, we must select the degree bound $\thresh$. More generally, we
are given a collection of possible extensions $\hat f_1,...,\hat f_k$,
each of which agrees with $f$ on a different set and %, and each of which
has different sensitivity $\Delta_i$.

For a large class of extensions, we can abstract the task we are faced
with as a private optimization problem: given a set of
\emph{real-valued} functions $q_1,...,q_k$, the goal is to output the
index $\hat \i$ of a function with approximately minimal value on the
data set $x$ (so that $q_{\hat\i}(x) \approx \min_i q_i(x)$). (In our
setting, the $q_i$ functions are related to the error of the
approximation $\hat f_i$ on $x$). Suppose that each $q_i$ has a known
(upper bound on) global sensitivity $\Delta_i$. The \emph{error} of an output
$\hat \i$ on input $x$ is the difference $q_{\hat\i}(x) - \min_i q_i(x)$.

The exponential mechanism of \citet{MT07}, a widely used tool in
differentially private algorithms, achieves error that scales with the
\emph{largest} of the sensitivities. Specifically, for every
$\beta>0$, with probability $1-\beta$, the output $\hat \i$ satisfies
$q_{\hat\i}(x)\leq \min_i q_i(x) + \Delta_{max}\cdot
\frac{2\log(k/\beta)}{\eps}$ where $\Delta_{max} = \max_i \Delta_i$.

In contrast, we give an algorithm whose accuracy scales with the
sensitivity of the \emph{optimal} score function $\Delta_{i^*}$ where
$i^*=\argmin_i q_i(x)$. Our mechanism requires as input an upper bound
$\beta>0$ on the desired probability of a ``bad'' outcome; the
algorithm's error guarantee depends on this $\beta$.

\begin{theorem}[Informal]\label{thm:genEM}
  For all settings of the input parameters $\beta\in (0,1)$, $\eps>0$, the
  Generalized Exponential Mechanism is
  $\eps$-differentially private. For all inputs $x$, the output $\hat \i$ satisfies
  \begin{displaymath}
  q_{\hat \i}(x) \leq \min_i \paren{ q_i(x) +  \Delta_i \cdot
  \tfrac{4\log(k/\beta)}{\epsilon} }\, .
\end{displaymath}
\end{theorem}

This guarantee can be much tighter than that
of the usual exponential mechanism. For instance, in our setting, the
$\Delta_i$'s grow exponentially with $i$ yet on sparse graphs, the
best choice of $\Delta_i$ is for $i$ relatively small. (Also, the
issue is not merely with the error guarantee. The exponential
mechanism provides bad outputs for many inputs where the true
minimizer has low sensitivity.)

We can use our algorithm for selecting the sensitivity parameter for the
Lipschitz extensions of graph functions in
\citep{BBDS13,KNRS13,ChenZ13} and in this work. (These parameters are
sometimes interpretable as a degree bound, as in the case of the
degree distribution, but not always; for example, when
computing the number of triangles, the parameter is a bound on the
number of triangles involving any one vertex). This allows the
algorithm to adapt to the specific input. The guarantee we get is
that the error of the overall algorithm (that is approximating some
function of an $n$-node graph) is at most $O(\log \log n)$ times higher than one
would get with the best Lipschitz constant.
In contrast, the
parameter selection method of \citet{ChenZ13} provides only a $O(\log
n)$ guarantee on the error blow-up, and is specific to the extensions
they construct.

% \paragraph{Bits and pieces on generalized exponential
%   mechanism.}Specifically, Lipschitz extensions provide approximations
% to the desired statistic. However, there is a challenge in picking,
% via a differentially private algorithm, the Lipschitz extension that
% will provide the most accurate result. While \cite{BBDS13} did not
% provide a theoretical tool for finding the right Lipschitz extension,
% \cite{KNRS13,ChenZ13} give a way to select it that results in loss in
% accuracy (compared to guessing the best extension) of a factor $O(\log
% n)$, where $n$ is the number of considered extensions (and also the
% number of nodes in the graph in the corresponding application to
% privacy).

\paragraph{Differentially Private Algorithms for Releasing the Degree
  Distribution\ifnum\full=1\  (Section~\ref{sec:decayerror})\fi.}

We can combine the Lipschitz extension of the degree list and the parameter selection
algorithm to get a differentially private mechanism for releasing the
degree distribution of a graph that automatically adapts to the
structure of the graph.

We show that our algorithm provides an accurate estimate on a large
class of graphs, including graphs with low average degree whose degree distribution is
heavy-tailed. We measure accuracy in the $\ell_1$ norm, normalized by
the number of nodes in the graph --- i.e., we
deem the algorithm accurate if the total variation distance between
the true degree distribution and the estimate is small.

This measure goes to 0 for graphs of low average degree in which the
tail of the degree distribution decreases slightly more quickly than
what trivially holds for all graphs.\snote{I am not sure this paragraph is clear, also need a ref?} If $\Av$ is the average degree
in a graph, Markov's inequality implies that the fraction of nodes
with degree above $t\cdot \Av$ is at most $1/t$. We assume that this
fraction goes down as $1/t^\alpha$ for a constant $\alpha>1$. The
condition is called \emph{$\alpha$-decay}. Our
algorithm need not be given $\alpha$ or the average degree of the
graph; these are implicitly taken into account by parameter
selection. Our assumption is satisfied by
all the well-studied social network models we know of, including
so-called~\emph{scale-free} graphs~\cite{CSN09}.

\ifnum\full = 0

\subsection{This Extended Abstract}

This short version of the document includes proofs of our two
main technical results: the Lipschitz extension of the degree list, in Section~\ref{sec:degree-list-extension}, and
the generalized exponential mechanism, in
Section~\ref{sec:newexpmech}.

The full version contains more background, as well as our
lower bounds on the stretch of extensions from $\Gd$ to $\G$ and
details of how to use the combine the main technical results to design
differentially private algorithms.

\fi

\ifnum\full=1 %%%MATCHES DEFINITIONS

\section{Definitions}\label{sec:definitions}
\mypar{Notation.}  We use $[n]$ to denote the set $\{1,\dots,n\}$. For a graph, $(V,E)$, $\Av(G) = 2|E|/|V|$ is the average degree of the graph $G$ and $\deg_v(G)$ denotes the degree of node $v \in V$ in $G$. When the graph referenced is clear, we drop $G$ in the notation. The asymptotic notation $O_n(\cdot), o_n(\cdot)$ is defined with respect to growing $n$. Other parameters are assumed to be functions independent of $n$ unless specified otherwise.

%Let $\G$ denote the set of unweighted, undirected finite \emph{labeled} graphs, and let $\G_n$ denote the set of graphs on at most $n$ nodes and $\G_{n,\thresh}$ be the set of all graphs in $\G_n$ with maximum degree $\thresh$.

\subsection{Graphs Metrics and Differential Privacy}
\label{sec:graphsDP}

\begin{defn}[$(\eps, \delta)$-edge/node-privacy]\label{def:dif-privacy}
A randomized algorithm $\A$ is \emph{$(\eps, \delta)$-edge-private} (respectively, \emph{node-private}) if for all events $S$ in the output space of $\A$, and edge (respectively, node) neighbors $G_1,G_2$,
$$\Pr[\A(G_1) \in S] \leq \exp(\eps)\times \Pr [\A(G_2) \in S]+\delta\,.$$
When $\delta=0$, the algorithm is $\eps$-edge-private (respectively, $\eps$-node-private). In this paper, if node or edge privacy is not specified, we mean node privacy by default.
\end{defn}

For simplicity of presentation, we assume that $n=|V|$, the number of nodes of the input graph $G$, is publicly known. This assumption is justified since, as we will see, one can get an accurate estimate of $|V|$ by running a node-private algorithm.

Both variants of differential privacy ``compose'' well, in the sense that privacy is preserved (albeit with slowly degrading parameters) even when the adversary gets to see the outcome of multiple differentially private algorithms run on the same data set.

\begin{lemma}[Composition, post-processing~\cite{MM09,DL09}]\label{lem:composition}
If an algorithm $\cala$ runs $t$ randomized algorithms $\cala_{1},\dots,\cala_{t}$, each of which is $(\eps,\delta)$-differentially private, and applies a randomized algorithm $g$ to the outputs, i.e., $\cala(G) = g(\cala_{1}(G),\dots,\allowbreak\cala_{t}(G)),$  then $\cala$ is $(t\eps,t\delta)$-differentially private. \drop{This holds even if for each $i>1$, $\cala_i$ is selected adaptively based on  $\cala_1(G),\ldots,\cala_{i-1}(G)$.}
\end{lemma}

\subsection{Basic Tools}\label{sec:sens}

\mypar{Global Sensitivity and the Laplace Mechanism.} In the most basic framework for achieving differential privacy, Laplace noise is scaled according to the {\em global sensitivity} of the desired statistic $f$. This technique extends directly to graphs as long as we measure sensitivity with respect to the metric used in the definition of the corresponding variant of differential privacy. Below, we explain this (standard) framework in terms of node privacy. Let $\G$ denote the set of all graphs.

\begin{definition}[Global Sensitivity~\cite{DMNS06}]\label{def:global-sensitivity}
The $\ell_1$-global node sensitivity of a function $f: \G \rightarrow \R^p$ is:
$$\displaystyle \Delta f = \max_{G_1,G_2\, \text{node neighbors}}\|f(G_1)-f(G_2)\|_1\,.$$
Equivalently, $\Delta f$ is the Lipschitz
constant of a function viewed as a map from $(\G,\dwire$) to
$\ell_1^p$.
\end{definition}

For example, the number of edges in an $n$-node graph has node
sensitivity $n$, since adding or deleting a node and its adjacent
edges can add or remove at most $n$ edges. In contrast, the number of
nodes in a graph has node sensitivity 1.

A {\em Laplace} random variable with mean $0$ and standard deviation $\sqrt{2}\lambda$ has density  $h(z)=(1/(2\lambda))e^{-|z|/\lambda}$. We denote it by
$\Lap(\lambda)$.

\begin{theorem} [Laplace Mechanism~\cite{DMNS06}]\label{thm:DMNS}
The algorithm $\cala(G)=f(G)+ \Lap(\Delta f/\eps)^p$ (which adds i.i.d.\ noise $\Lap(\Delta f/\eps)$ to each entry of $f(G)$) is $\eps$-node-private.
\end{theorem}

Thus, we can release the number of nodes $|V|$ in a graph with noise of expected magnitude $1/\eps$ while satisfying node differential privacy. Given a public bound $n$ on the number of nodes, we can release the number of edges $|E|$ with additive noise of expected magnitude $n/\eps$.

% \mypar{Local Sensitivity.} The magnitude of noise added by the Laplace mechanism depends on $\Delta f$ and the privacy parameter $\eps$, but not on the database $G$. For many functions, this approach results in high noise, not reflecting the function's typical insensitivity to individual inputs. Nissim~\etal~\cite{NRS07} proposed a local measure of sensitivity, defined next.

% \begin{definition}[Local Sensitivity~\cite{NRS07}]\label{def:local-sens}
% For a function $f:\G \to\mathbb{R}^p$ and a graph $G\in \G$, the local sensitivity of $f$ at $G$ is
% $
% \displaystyle
% \lsf{G}{f} = \max_{G_2:\text{node neighbor of $G_1$}} \|f(G_1)-f(G_2)\|_1.$
% \end{definition}
% \noindent Note that, by Definitions~\ref{def:global-sensitivity} and~\ref{def:local-sens}, the global sensitivity $\Delta f=\max_{G\in\G} \lsf{G}{f} $. One may think of the local sensitivity as a discrete analogue of the magnitude of the gradient of~$f$.

% A straightforward argument shows that every differentially private
% algorithm must add distortion at least as large as the local
% sensitivity on many inputs. However, finding algorithms whose error
% matches the local sensitivity is \emph{not} straightforward:  an
% algorithm that releases $f$ with noise magnitude proportional to
% $\lsf{G}{f} $ on input $G$ is not, in general, differentially
% private~\cite{NRS07}, since the noise magnitude itself can leak
% information.

\mypar{Exponential Mechanism.}

Suppose data sets are members of a universe $U$ equipped with a neighbor
relation (for example, $U=\G$ with
vertex neighbors).
Suppose we are given a collection of functions $q_1,...,q_k$, from $U$
to $\R$ such that for
each $i\in [k]$, the function $q_i(\cdot)$ has sensitivity at most $\Delta$. The
exponential mechanism  (\citet{MT07}) takes a data set and aims
to output an index $\hat\i$ for which $q_{\hat \i}(G)$ has nearly minimal value
at $G$, that is, such that $q_{\hat \i}(G) \approx \min_i q_i(G)$. The
algorithm $\A$ samples an index $i$ such that $\Pr(\A(G)=i)
\propto \exp\paren{\frac{\eps}{2\Delta}q_i(G)}\,.$

\begin{lemma}[Exponential Mechanism \citep{MT07}]
\label{lem:exp-mech}
The algorithm $\A$
  is $\eps$-differentially private. Moreover, with probability at
  least $1-\eta$, its output $\hat\i$ satisfies $q_{\hat \i}( G) \leq
  \min_i\paren{q_i(G)} + \frac{2 \Delta \ln(k/\eta)}{\eps}\,. $
\end{lemma}

There is a simple, efficient implementation of the exponential
mechanism that adds exponential noise to each score function and
reports the maximizer of the noisy scores~(see, e.g., \citep[Sec. 3.4]{DworkRoth14book}).

\section{General Results on Lipschitz Extensions}
\label{sec:general-lip-ext}
A number of basic results from functional analysis apply to
our setting. Let $\ell_p^d$ denote the set $\re^d$ equipped with the $\ell_p$
metric.

When $Y=\re$ (with the usual metric), a Lipschitz extension always
exists~\cite{McShane34}. The classic construction, given a $c$-Lipschitz function
$f:X\to\re$, defines $\hat f:X'\to \re$ as
$$\hat f(y) = \inf_{x\in X} \paren{f(x) + c\cdot d_{X'}(x,y)}\,.$$
The function $\hat f$ is also $c$-Lipschitz, but need not necessarily
be easy to compute even if $f$ admits efficient algorithms.

\citet{BBDS13,KNRS13,ChenZ13} constructed \emph{polynomial-time} Lipschitz extensions from $\Gd$ to
$\G$ of several real-valued functions on graphs (see Introduction).

In this work, our focus is on higher-dimensional functions on graphs,
i.e., functions that map graphs into $\R^p$ for $p>1$.
As with one-dimensional functions, there always exist stretch-1
extensions of functions that take values in $\ell_\infty^p$ for any
dimension $p$, since one can separately find an extension for each
coordinate of $f$. It is also true for $\ell_1^2$, since $\ell_1^2$ is
isomorphic to $\ell_\infty^2$.  However, stretch-1 extensions need not
exist when $Y=\ell_2^p$ or $\ell_1^p$ for larger $p$. There is a
growing body of theory on the minimal stretch required for extensions
among different spaces; see \cite{BenyaminiL98book,LeeN05} for a concise summary of
known general results on the problem.

Our first result is that one cannot always get stretch-1 extensions
for functions from $\Gd$ to $\ell_1^p$ or $\ell_2^p$. We prove it at
the end of this section. It is the only lower bound on extendability for
these metrics we are aware of.

\begin{prop}\label{prop:constantstretch}
Consider the vertex distance on $\G$. There is an
  absolute constant $c>1$ such that: (1) for all $p\geq 3$, there exist
  symmetric functions from $\Gd$ to $\ell_1^p$ that do not admit a
  stretch-$c$ extension to $\G$; (2) for all $p\geq 2$, there exist
  symmetric functions from $\Gd$ to $\ell_2^p$ that do not admit a
  stretch-$c$ extension to $\G$.
\end{prop}

This lower bound extends to edge distance on $\G$ (we omit the
proof). Moreover, for edge distance, it is essentially tight: a result of \citet{BBDS13} on
smooth projections implies that every function on $\Gd$
which is Lipschitz under the \emph{edge distance} metric on $\Gd$
can be extended to all of $\G$ with stretch at most $3$, regardless
of the output metric. However, the construction does not apply to
vertex distance on graphs.

For the \emph{vertex distance} on $\Gd$, known results yield
extensions with stretch that is polynomial in either $p$ or $n$ (the
size of the graph). We outline these briefly: \citet[Theorem
1.6]{LeeN05} show that one can get extensions with stretch
$O(\rho(X))$, where $\rho(X)$ is the \emph{doubling dimension} of the
metric space $X$ (in our case, $\Gd$ or $\Gnd$).  Unfortunately, the
vertex metric on $\Gnd$ has doubling dimension at least $n$, even for
$\thresh=4$ and even if we identify isomorphic graphs (see
Appendix~\ref{sec:doubling} for formal definitions and a
proof). \citet{MakarychevM10} show that functions from any metric on
$N$ points can be extended to an arbitrary containing space with
stretch at most $O(\log N / \log \log N)$. Since $\log N$ is
approximately $n\thresh$ for $\Gnd$, this again yields large stretch.
Finally, another general approach, based on the dimension of the image
space, yields stretch $p$ and $\sqrt{p}$ for maps into $\ell_1^p$ and
$\ell_2^p$ respectively (in our case, one can obtain this by
separately extending each of the $p$ coordinates of the output).

\begin{proof}[Proof of Proposition~\ref{prop:constantstretch}] Our
proof is inspired by the example of \citet{BenyaminiL98book} of spaces
$X\subseteq X'$ and a function $f:X\to\ell_2^2$ such that there is no
stretch-1 extension of $f$ from $X$ to $X'$.

We start with the case of maps into $\ell_1$.
  Let $X'$
  denote the metric space $\{a,b,c,d,e\}$ with all pairwise distances
  among $X=\{a,b,c,d\}$ equal to 2, and distances $d_{X'}(x,e)=1$ for
  $x\in X$ (pictured as a graph below).
  Consider the function $f:X\to
  \re^3$ that maps $X$ to the corners of a particular tetrahedron:
  \begin{equation}
\vcenter{\hbox{\includegraphics[height=1in]{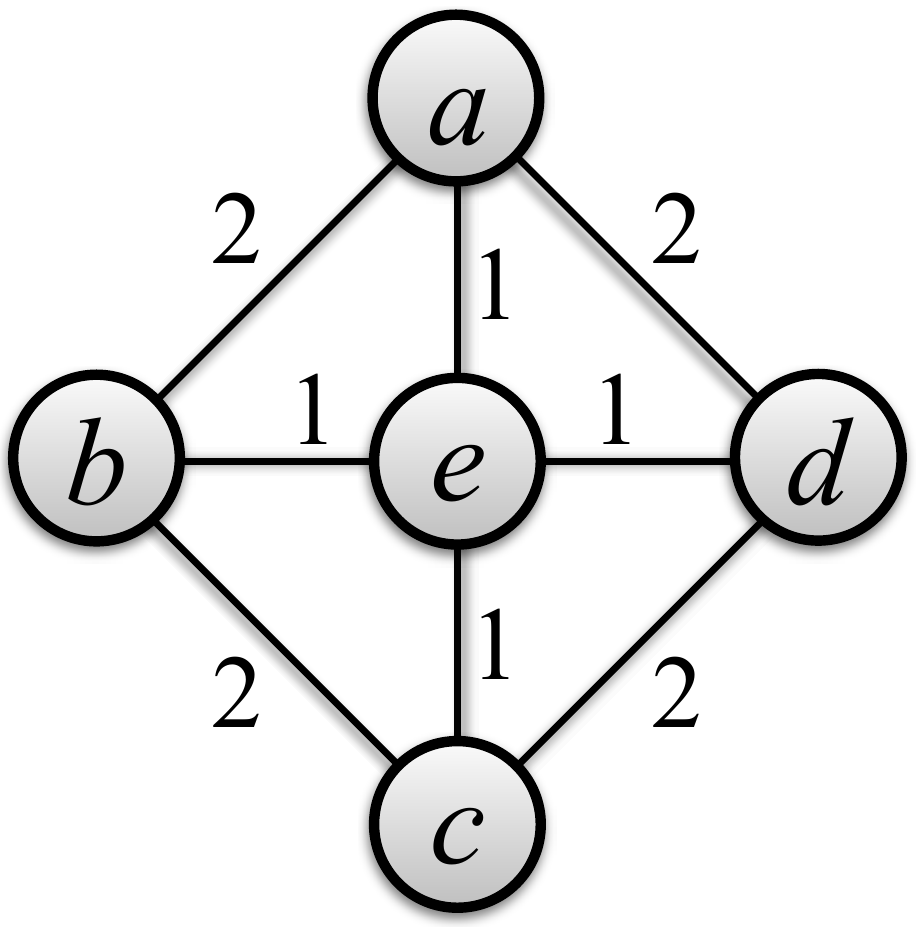}}}
\qquad \qquad \begin{array}{cccrrrc}
  f(a)&=&(&-1,&1,&1&) \\
  f(b)& =& (&1,&-1,&1&)\\
  f(c) &=&(&1,&1,&-1&) \\
  f(d) &=&(&-1,&-1,&-1&) \\
\end{array}\label{eq:counterex}
\end{equation}
The function $f$ is 2-Lipschitz if we view the image as $\ell_1^3$, but there is
no way to extend it to all of $X'$ in either metric without stretching
the Lipschitz constant. To satisfy the Lipschitz
constraint $f(e)$ has to be exactly halfway between every pair in the
set $\{f(a),f(b),f(c), f(d)\}$ (since it has to be at distance at most
2 from each of the points). The points that are halfway from $a$ to
$b$ have third coordinate 1; the points halfway from $c$ to $d$ have
third coordinate -1; there is no intersection between the two sets,
and hence no possible value for $f(e)$. Any value for $f(e)$ results
in a stretch of at least some absolute constant $c>1$.

We can lift this example to other domains $X\subset X'$. For example, we can take $X'$
to be $\ell_1^4$, and let $a=(1,0,0,0), b=(0,1,0,0), c=(0,0,1,0),  d= (
(0,0,0,1) $ and $e=
(0,0,0,0)$.

Lifting the example to $\Gd\subset \G$ is a bit messier. Fix $d$ at least
4. Let $G_0$ be a graph on at least $4(d-2)$ vertices with maximum
degree at most $d-1$ and no nontrivial automorphisms (a sufficiently
large random graph satisfies the criteria with high
probability~\cite[Chap. 9]{Bollobas98book}). We create a larger graph
$H$ by adding four vertices $\{t,u,v,w\}$ to $G_0$, among which all
possible edges exist, and such that $t,u,v$ and $w$ are connected to a
disjoint subsets of $d-2$ vertices in $G_0$ (this is possible since $G_0$
must have at least $4(d-2)$ vertices). The vertices $t,u,v,w$ have
degree $d+1$ in $H$.
\begin{center}
  \includegraphics[height=1.3in]{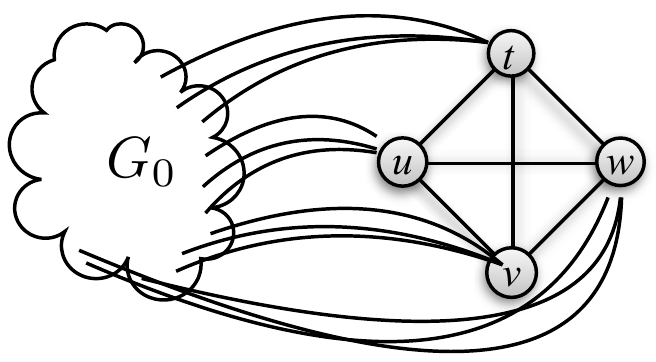}
\end{center}
To embed our counterexample in $\G$, let $e=H$, and let $\{a,b,c,d\}$
be the graphs obtained by deleting one of $t,u,v,w$ (respectively)
from $H$. The four graphs ${a,b,c,d}$ lie in $\Gd$, and no pair of
them is isomorphic (since $u,v,w,$ are connected to disjoint sets of a
graph with no automorphisms). The vertex distance between any pair of
graphs in $a,b,c,d$ is 2, and their distance from $e$ is 1. We can set
the values of $f$ on $a,b,c,d$ as in \eqref{eq:counterex} (this is
consistent with the requirement that $f$ be symmetric since the graphs
are not isomorphic). By the reasoning above, $f$ is 2-Lipschitz but
there is no way to assign a value to $f(e)$ without increasing the
stretch of $f$.

We must still show that it is possible to assign values to functions on
the remaining graphs in $\Gd$ without increasing the Lipschitz
constant.

For the graphs $G$ that can be obtained by exactly two of $t,u,v,w$
(along with corresponding edges) to $G_0$, there two of the graphs in
$\{a,b,c,d\}$ that are at distance 1 from $G$. We set $f(G)$ to be the
average of the values of $f$ at these two nearest graphs (for example,
$G+\{t,u\}$ is at distance 1 from graphs $c$ and $d$; we set
$f(G+\{t,u\})=(0,0,-1)$. Note that $f( G+\{t,u\})$ is at distance 2
from $f(c)$ and $f(d)$, as required. For all other graphs, $G\in \Gd$,
we set $f(G) = (0,0,0)$. One can verify by inspection that the
2-Lipschitz property is satisfied on all of $\Gd$ by $f$.

Finally, we note that an even simpler example works for maps into
$\ell_2^2$. Starting with the same spaces $X$ and $X'$, we can
consider a function $f:X\to\re^2$ that maps $\{a,b,c\}$ to the corners
of an equilateral triangle with side-length 1. The map is $\frac 1 2$ -Lipschitz
on $\{a,b,c\}$, but cannot be extended to all of $X'$ (since there is
no point at distance $\frac 1 2$ of all three corners. Lifting the
example to $\Gd\subset \G$ is similar to the $\ell_1$ case.
\end{proof}

\fi %%%MATCHES DEFINITIONS

\ifnum\full=1
\section{Lipschitz Extensions of the Degree List and Distribution}
\label{sec:degree-list-extension}

\subsection{Lipschitz Extension of the Degree List}
\else
\section{Lipschitz Extensions of the Degree List}
\label{sec:degree-list-extension}
\fi

In this section, we give a Lipschitz extension of the degree list. For an $n$-node graph $G$, let
$$\deglist(G) = sort(deg_1(G),...,deg_{|V(G)|}(G))$$ denote the list of degrees
of $G$ sorted in nonincreasing order.

We view the degree list as an element of $\re^*$ (the set of finite
sequences of real numbers). We equip the space with the $\ell_1$
distance, where the sequences of different lengths are padded with 0's
to allow comparison. This representation is
convenient for handling node additions and deletions.

The global $\ell_1$ sensitivity (under node insertion and removal) of
the degree list on $\thresh$-bounded graphs is $2\thresh$ because the
unsorted degree list has sensitivity $2\thresh$ and, as  \citet{HayLMJ09} observed, sorting does not
increase the $\ell_1$ distance between vectors. We construct an
extension that agrees with $\deglist$ on $\Gd$ and has global
sensitivity at most $3\thresh$.

\anote{Added an explanation of why naive constructions don't work.}
Before explaining our construction, we consider a simpler ``straw man''
construction to illustrate the problem's difficulty: suppose that
given the degree list $\deglist(G)$, we obtain $\hat f_\thresh(G)$ by
rounding all degrees above $\thresh$ down to $\thresh$. This will not
affect the degrees in graph with maximum degree $\thresh$, but it is
not $O(\thresh)$ Lipschitz: consider a star graph on $n$ vertices,
with one vertex of degree $n-1$ and $n-1$ vertices of degree 1. Simple
rounding would report $\hat f(G)$ as $(\thresh, 1,....,1)$. But the
graph has a neighbor $G'$ with no edges at all, for which the reported
degree list would be all 0's. Those vectors differ by $n+ \thresh-1$in
the $\ell_1$ norm. One can try simple ways of dropping very
high-degree vertices (an idea called ``projection'' in
\citep{BBDS13,KNRS13}), but those do not yield uniform bounds on the
sensitivity of the resulting degree sequence and result in more noise
being added for privacy.

Like in~\cite{KNRS13}, our starting point is the construction of the  flow graph $G'$ for graph $G$. \citet{KNRS13} proved that the value of the maximum flow in $G'$ is a Lipschitz extension of the number of edges in $G$. We will use the flow values on certain edges as a proxy for degrees of related vertices. The main challenge is that, whereas the value of the maximum flow in $G'$ is unique, the actual flow on specific edges is not.

\begin{definition}[Flow graph]\label{def:flow-graph}
Given a graph $G=(V,E)$, let $V_\ell=\{v_\ell \mid v\in V \}$ and $V_r=\{v_r \mid v\in V \}$ be two copies of $V$, called the {\em left} and the {\em right} copies, respectively. Let $\thresh$ be a natural number less than $n$.  The {\em flow graph} of $G$ with threshold $\thresh$, a source $s$ and a sink $t$ is a directed graph on nodes $V_\ell\cup V_r \cup \{s,t\}$ with the following capacitated edges: edges of capacity $\thresh$ from the source $s$ to all nodes in $V_\ell$ and from all nodes in $V_r$ to the sink $t$, and unit-capacity edges $(u_\ell,v_r)$ for all edges $(u,v)$ of $G$.
The flow graph of $G$ is denoted $\FG(G)$.
%Let $\fl(G)$ denote the value of the maximum flow in the flow graph of $G$.
\end{definition}

We would like our extension function to output the sorted
list of flows leaving the source vertex in some maximum flow. The
challenge is that there may be many maximum flows. If
we select a maximum flow arbitrarily, then the selected flow may be very
sensitive to changes in the graph, even though its value changes
little. We get around this by selecting a flow that
minimizes a
strictly convex function of the flow values.
%\anote{Potential alternate version. Old version at rev 5793}

\begin{definition}[Lipschitz extension of degree list]\label{def:lip-ext-degree-list}
Given a flow $f$ of $\FG(G)$, let $f(e)$ denote the flow on an edge $e$. Also, let $\outfl{f}$ be the vector of
flows on the edges leaving the source $s$, let $\infl{f}$ be the
vector of flows on the edges entering $t$, and let $\inoutfl{f}$ be the concatenation of the two vectors. We use $\Dvec$ to denote a vector of length $2n$, where all entries are $D$.
Let $\Phi(f)$ be the squared
$\ell_2$ distance between $\inoutfl{f}$ and $\Dvec$, that is,
$$\Phi(f) = \| \inoutfl{f} - \Dvec\|_2^2 = \sum_{v\in V}\paren{ (\thresh - f(s,v_{\ell}))^2 + (\thresh - f(v_{r}, t))^2}\, .$$
Let $f$ be the flow that minimizes the objective function $\Phi$ over all feasible flows in $\FG(G)$.
Define $\dlext(G)$ to be the sorted list of flows along the edges
leaving the source, that is, $\dlext(G)=sort(\infl{f})$.
\end{definition}

The function
$\dlext(G)$ is uniquely defined because the objective $\Phi$ is strictly convex
in the values $\inoutfl{f}$.
$\dlext(G)$ can be approximated to arbitrary precision in polynomial time, since it is the
minimum of a strongly convex function over a polytope with
polynomially many constraints. The
approximation may slightly increase the sensitivity; in our
application, one can account for this by adding slightly more than
$3\thresh/\eps$ noise in each coordinate.
\anote{What is the actual best known running time? I couldn't find it
  in \cite{LeeRS13}, which is the reference I was given.}

Theorem~\ref{thm:lip-ext-degree-list-intro} follows from the following theorem.
\begin{theorem}\label{thm:lip-ext-degree-list}
The function  $\dlext(G)$ is a Lipschitz extension of $\deglist(G)$ from $\Gd$ to $\G$ of stretch~3/2.
In other words,
\begin{enumerate}
\item   If $G$ is $\thresh$-bounded, then $\dlext(G)=\deglist(G)$.
\item   For any two graphs $G_1,G_2$ (not necessarily  $\thresh$-bounded) that are node neighbors, $$\|\dlext(G_1)-\dlext(G_2)\|_1 \leq 3\thresh\,.$$

\end{enumerate}
\end{theorem}

\begin{proof}[Proof of Theorem~\ref{thm:lip-ext-degree-list} (item 1)]
  The flow that assigns 1 to all edges $(u_\ell, v_r)$ and $deg(v)$ to all edges $(s,v_\ell)$ and $(v_r,t)$
  strictly dominates all feasible flows. In particular, it minimizes $\Phi$
  since, for $x\in[0,D]$, function $(D-x)^2$ is decreasing in $x$.
\end{proof}

There are two distinct notions of optimality of a flow in $\FG(G)$: optimality with respect to $\Phi$, which we call \emph{$\Phi$-optimality}, and optimality of the net flow form $s$ to $t$, called \emph{net flow optimality}. Next, we show that $\Phi$-optimality implies net flow optimality.

\begin{lemma}\label{lem:optimality}
  For every graph $G$, if $f$ minimizes $\Phi$ among valid flows for the flow graph $\FG(G)$,
  then $f$ has maximum net flow from $s$ to $t$ in $\FG(G)$.
\end{lemma}

\begin{proof}
  If $f$ does not have maximum net flow, then we can find a shortest augmenting
  path $p$ from $s$ to $t$. Let  $c>0$ be the minimal residual capacity of the edges in $p$. Since $p$ is a
  shortest path, it is simple; thus, adding $cp$
  to $f$ results in a feasible flow, but does not decrease the flow along any edge leaving $s$ or entering $t$. This implies that $\Phi(f+cp)<\Phi(f)$ (since $\Phi$ is strictly decreasing in each argument), contradicting the $\Phi$-optimality of $f$.
\end{proof}

The flow graph $\FG(G)$ admits a simple symmetry: for any flow $f$, we can obtain a feasible flow $\pi(f)$ by swapping the roles of $s$ and $t$ and the roles of left and right copies of all vertices. That is, we define $\pi(f)(s,v_\ell):= f(v_r,t)$,
$\pi(f)(u_r,t):= f(s,u_\ell)$, $\pi(f)(u_\ell,v_r):= f(v_\ell,u_r)$ for all vertices $v,u$ in $G$.  Flow $f$ is \emph{symmetric} if $\pi(f) = f$. For every graph $G$, there exists a symmetric $\Phi$-optimal
  flow in $\FG(G)$: given any $\Phi$-optimal flow $f'$, the flow $f'' = \frac 1 2 (f' + \pi(f'))$ is symmetric, feasible (because the set of feasible flows is convex) and has objective value at most $\Phi(f')$ by convexity of~$\Phi$.

\begin{proof}[Proof of Theorem~\ref{thm:lip-ext-degree-list} (item 2)]
  Suppose a graph $G_1$ on $n$ vertices is obtained by removing a node
  $\vnew$ along with its associated edges from a graph $G_2$ (on $n+1$
  vertices).

Let $f_1, f_2$ be $\Phi$-optimal symmetric flows for the flow
  graphs $\FG(G_1)$ and $\FG(G_2)$, respectively.

Observe that $f_1$ is a feasible flow in $\FG(G_2)$.  Consider the flow $\Delta = f_2-f_1$. Note that   $\Delta$
  is a maximum signed\snote{Why do we need ``signed''?}\anote{Because
    some entries of $f_2-f_1$ can be negative since some flow can be re-routed.} flow in the residual graph of flow
  $f_1$ for $\FG(G_2)$. In particular, $\Delta$ satisfies flow and capacity constraints,
  but not necessarily positivity. Since $\|\dlext(G_1)-\dlext(G_2)\|_1=\|\outfl{\Delta}\|_1$, our goal is to prove $\|\outfl{\Delta}\|_1\leq 3\thresh.$

Next, we decompose $\Delta$ into three {\em subflows}. A {\em subflow} of a flow
$\Delta$ is a flow $\Delta'$ such that for all edges $e$, the flows $\Delta(e)$ and $\Delta'(e)$ cannot have different signs and $|\Delta'(e)|\leq \Delta(e)$. We start by decomposing $\Delta$ into subflows that form simple $s$-$t$ paths and simple cycles. Then we group them as follows:
\begin{itemize}
\item Let $\Delta^s$ be the sum of all flows from the initial decomposition that form paths and cycles using the edge $(s,\vnew_\ell)$.
\item Let $\Delta^t$ be the sum of all flows from the initial decomposition that form paths and cycles using the edge $(\vnew_r, t)$, but not $(s,\vnew_\ell)$.
\item Let $\Delta^0$ be the sum of the remaining flows, i.e., $\Delta^0=\Delta-\Delta^s-\Delta^t.$
\end{itemize}
Since, by definition of the subflow decomposition, $\|\outfl{\Delta}\|_1=\|\outfl{\Delta^s}\|_1+\|\outfl{\Delta^t}\|_1+\|\outfl{\Delta^0}\|_1,$ it remains to bound
 the three values in the sum. We do it in the following three lemmas.

\begin{lemma}\label{lem:delta-s}
$\|\outfl{\Delta^s}\|_1\leq 2D.$
\end{lemma}
\begin{proof}
Recall that $\Delta^s$ can be decomposed into simple $s$-$t$ paths and simple cycles that use the edge $(s,\vnew_\ell)$. Each such path contributes the value of its flow to $\|\outfl{\Delta^s}\|_1$, and each such cycle contributes at most twice the value of its flow. Since the total flow $\Delta^s(s,\vnew_\ell)$ is at most $D$, we get that $\|\outfl{\Delta^s}\|_1\leq 2D.$
\end{proof}

\begin{lemma}\label{lem:delta-t}
$\|\outfl{\Delta^t}\|_1\leq D.$
\end{lemma}
\begin{proof}
Recall that $\Delta^t$ can be decomposed into simple $s$-$t$ paths and cycles that use the edge $(\vnew_r,t),$ but not $(s,\vnew_\ell)$. Each such path contributes the value of its flow to $\|\outfl{\Delta^s}\|_1$.  Any such cycle contributes 0 to $\|\outfl{\Delta^s}\|_1$ because any simple cycle in $\Delta$ that starts from $t$
  cannot reach $s$. If it did, one could find an augmenting $s$-$t$
  path in $\Delta$, implying that $f_2$ is not a net value optimal flow in $\FG(G_2)$ and, by Lemma~\ref{lem:optimality}, contradicting $\Phi$-optimality of $f_2$ in $\FG(G_2)$.

Since the total flow $\Delta^t(\vnew_r,t)$ is at most $D$, we get that $\|\outfl{\Delta^t}\|_1\leq D.$
\end{proof}

\begin{lemma}\label{lem:delta0}
$\|\outfl{\Delta^0}\|_1=0.$
\end{lemma}
\begin{proof}
The flow  $\Delta^0$ does not use the edges $(s,\vnew_\ell)$ and $(\vnew_r,t)$ since all flow in $\Delta$ along $(s,\vnew_\ell)$ and $(\vnew_r,t)$ has been used by $\Delta^s+\Delta^t$. Consequently, $\Delta^0$ has no flow passing through $\vnew_\ell$ and $\vnew_r$. Therefore, $\Delta^0$ is a feasible flow for the residual graph of $f_1$ in
  $\FG(G_1)$. We conclude that $f_1+\Delta^0$ is feasible in $\FG(G_1)$.

Suppose for the sake of contradiction that $\|\outfl{\Delta^0}\|_1>0.$ Then we can use convexity of $\Phi$ to prove the following inequalities:
\begin{eqnarray}
\label{eq:delta0-1}\langle \Delta^0,\Dvec-f_1\rangle &\leq& 0.\\
\label{eq:delta0-2}\langle \Delta^0,\Dvec-(f_2-\Delta^0)\rangle &>& 0.
\end{eqnarray}
To prove (\ref{eq:delta0-1}), consider the polytope of feasible flows in $\FG(G_1)$. Both $f_1$ and $f_1+\Delta^0$ are in the polytope. Moreover, $f_1$ is the unique $\Phi$-optimal flow in $\FG(G_1)$. Since $\Phi$ is minimized at $\Dvec$, a tiny step from $f_1$ in the direction of $f_1+\Delta^0$ takes us further from $\Dvec$. In other words, the angle between the vectors $(f_1,\Dvec)$ and $(f_1,f_1+\Delta^0)$ is at least $90^\circ$, implying (\ref{eq:delta0-1}).

To prove (\ref{eq:delta0-2}), consider the polytope of feasible flows in $\FG(G_2)$. Both $f_2$ and $f_2-\Delta^0$ are in that polytope. Moreover, $f_2$ is the unique $\Phi$-optimal flow in $\FG(G_2)$. Since $\Phi$ is minimized at $\Dvec$, a tiny step from $f_2-\Delta^0$ in the direction of $f_2$ takes us closer to $\Dvec$. In other words, the angle between the vectors $(f_2-\Delta^0,f_2)$ and $f_2-\Delta^0,\Dvec)$ is less than $90^\circ$, implying (\ref{eq:delta0-2}).

Subtracting (\ref{eq:delta0-2}) from (\ref{eq:delta0-1}) and using the fact that $\Delta=f_2-f_1=\Delta^s+\Delta^t+\Delta^0$, we get
\begin{eqnarray}
\nonumber \langle \Delta^0,\Dvec-(f_2-\Delta^0)\rangle -\langle \Delta^0,\Dvec-f_1\rangle &>& 0;\\
\nonumber\langle \Delta^0,-(f_2-f_1-\Delta^0)\rangle &>&0;\\
\label{eq:delta0-contradiction}\langle \Delta^0,\Delta^s+\Delta^t\rangle &<&0.
\end{eqnarray}
But $\Delta^0$ and $\Delta^s+\Delta^t$ are both subflows of $\Delta$, so they  cannot have opposite signs, on any edge, contradicting (\ref{eq:delta0-contradiction}). Therefore, $\|\outfl{\Delta^0}\|_1=0.$
\end{proof}

We now complete the proof of Theorem~\ref{thm:lip-ext-degree-list} (Item 2). Recall that $\Delta=\Delta^s+\Delta^t+\Delta^0$ and that $\Delta^s, \Delta^t,$ and $\Delta^0$ are subflows of $\Delta$. From Lemmas~\ref{lem:delta-s}--\ref{lem:delta0}, we get
$\|\dlext(G_1)-\dlext(G_2)\|_1=\|\outfl{\Delta}\|_1 = \|\outfl{\Delta^s}\|_1+ \|\outfl{\Delta^t}\|_1+\|\outfl{\Delta^0}\|_1\leq 3\thresh,$
as desired.
\end{proof}

\ifnum\full = 1 %%% MATCHES DL DETAILS

\subsection{From the Degree List to the Degree Distribution}
\label{sec:DLtoDD}

Let $p_G$ denote the degree distribution of the graph $G$, i.e.,
$p_G(k) = \big|\{v: \deg_v(G) =k\}\big | / |V|$. Similarly, $\cum_G$
denotes the {\em cumulative} degree distribution (CDF), i.e., $\cum_G(k) =
\big|\{v: \deg_v(G) \geq k\}\big| / |V|$.

We can modify the extension of the degree list to get extensions of
the \emph{degree histogram} $n\cdot p_G$ or the cumulative degree
histogram (CDH)
$n\cdot P_G$. If we consider two \emph{integral} degree lists that are
at $\ell_1$ distance $t$, then
the $\ell_1$ distance between their CDH's is
at most $t$ (similarly for degree histograms). However, since our extension
of the degree list may produce fractional lists, we need to extend the
CDH to fractional degree lists so that the map
from lists to CDHs remains
Lipschitz in the $\ell_1$ norm.

We do this first for the CDH; the extension of the degree
histogram is an easy modification.
Given an integer $k\in [\thresh]$, let
$$[x]_k=\max\{0, \min\{1,x-(k-1)\}\} =
\begin{cases}
  0& \text{if }x\leq k-1,\\
  x-(k-1) & \text{if }k-1\leq x \leq k,\\
  1& \text{if }x \geq k.\\
\end{cases}
$$
Define the map $H$ as follows: for a nonnegative real number $a$,
$H(a)=([a]_1,[a]_2,...,[a]_{\lceil{a}\rceil})$. (This is a vector of
length $\lceil{a}\rceil$ whose $\ell_1$ norm is exactly $a$.) Given a finite sequence $(a_1,...,a_n)\in
[0,\thresh]^*$, let $H(a_1,...,a_n) = \sum_i H(a_i)$, where we pad
shorter sequence with 0's to allow summation. If the input numbers are
in $[0,\thresh]$, the sequence has length at most $\thresh$.

\begin{lemma}
  The function $H$ is 1-Lipschitz in the $\ell_1$ norm. That is, $\|H(a) -
  H(a') \|_1\leq \|a-a'\|_1$ for all vectors $a,a'\in
  [0,\thresh]^*$. Moreover, for every graph $G$, $H(\deglist(G)) =
  n\cdot P_G$ where $n=|V_G|$.
\end{lemma}

\begin{proof}
  This follows from the fact that $H(a)$ has $\ell_1$ norm $a$ for
  every nonnegative real number,  and equals the sequence $1^a$ when $a$ is an integer.
\end{proof}

Given $H$, which extends transforms degree lists to
the CDH, we can obtain an extension of the degree histogram via
$hist_\thresh(a) = H_\thresh(a)$, and $hist_i(a) = H_i(a) - H_{i+1}(a)$ for
$i<\thresh$. This increases $\ell_1$ distances by at most an
additional factor of 2.

\begin{theorem}
  The map
$$G\mapsto hist(\hat f_\thresh(G))$$
extends the degree histogram (as a map from $(\Gd,\dwire)$ to
$\ell_1^\thresh$) to $\G$, with stretch at most
3.
\end{theorem}
\subsection{Differentially Private Approximations to the Degree
  Distribution}
\label{sec:DPdd}
There are two natural approaches to using the extension of $\deglist$
to release an approximate degree distribution. First, we may add noise
$\thresh/\eps$ to each entry of the sorted degree list, and project
(and/or) remove noise as in \cite{HayLMJ09,KarwaS12,LinKifer13}. The
second is to
release the $\thresh$-bounded degree histogram and add noise. The
error of the first approach is difficult to bound analytically, and so
we adopt the second here.

Given a degree threshold $\thresh$, consider the following mechanism:

\begin{algorithm}[h]\caption{Noisy Degree Histogram$(G,\eps,\thresh)$}
  $Y_i \sim \Lap(6\thresh / \eps)$ for $i=1,...,\thresh$\;
  \Return $\A_\thresh(G) = hist(\hat f_\thresh(G)) + (Y_1,...,Y_\thresh) $\;
\end{algorithm}
(We only need to release the first $\thresh$ entries of $hist$,
since the remaining entries are always 0.)

This mechanism introduces two sources of error: the extension error
$\hat f(G) - \deglist(G)$ and the random noise $\vec Y =
(Y_1,...,Y_d)$. The noise component is easy to understand and
bound. How can we characterize the error introduced by the extension?

\begin{lemma}\label{lem:l1extensionerror} For any graph $G$ and threshold $\thresh$, the extension's $\ell_1$ error
   satisfies
$$n \sum_{i> D} P_G(i)\leq \|\dlext (G) -\deglist(G)\|_1\leq 2n
\sum_{i> D} P_G(i) \, .$$
\end{lemma}
\begin{proof}
  Recall that $\hat f(G)$ has minimal $\ell_1$ error $\|\dlext
  (G) - \deglist(G)\|_1$ among all $\thresh$-bounded real vectors that
  are consistent with a weighted graph. In particular, one can
  consider a graph $G'$ which is obtained by removing $\deg_v
  -\thresh$ edges for each vertex $v$ with degree greater than
  $\thresh$. Each edge removal causes an change of 2 in
  $\deglist(G)$ in the $\ell_1$ norm. The number of edges removed is
  $\sum_{v:\ \deg_v> \thresh} (\deg_v-\thresh)\, .$
  An alternative formula for this sum can be obtained by summing over
  degrees instead of vertices:
  $$\sum_{v:\ \deg_v> \thresh} (\deg_v-\thresh) = \sum_{i>\thresh} n P_G(i)
  $$
  (since each vertex $v$ contributes $\max(0,\deg_v-\thresh)$ to the
  sum).  Multiplying by 2 yields the desired upper bound.

  To prove the lower bound, note that the vector $\hat
f_\thresh(G)$ is always less, coordinatewise, than the simple
projection that replaces the degree $\deg_v$ of each vertex $v$ by
$\min(\thresh,\deg_v)$. The $\ell_1$ error of $\hat f(G)$ (or indeed
of any function that projects onto a set of vectors with entries
bounded by $\thresh$) is therefore
at least $\sum_{v:\ \deg_v> \thresh} (\deg_v-\thresh) =
n\sum_{i>\thresh} P_G(i) $.
\end{proof}

Combining the two previous lemmas with the fact that the expected absolute
value of each $Y_i$ is $4\thresh/\eps$, we obtain the following theorem.

\begin{theorem}\label{thm:ell1errorfixedthreshold}
  The expected $\ell_1$ error of algorithm   $\A_\thresh$ on input $G$ is
  at most $\displaystyle 2 n \sum_{i> D} P_G(i) + \frac{6\thresh^2 }{\epsilon} \, .$
\end{theorem}

This theorem bounds the error of the algorithm for a given degree
threshold $\thresh$. In the sequel, we show how we can select a
(nearly) optimal threshold differentially privately.

\fi %%% MATCHES DL DETAILS

\section{Exponential Mechanism For Scores With Varying Sensitivity}
\label{sec:newexpmech}

The exponential mechanism of \citet{MT07} is a basic tool for designing
differentially private algorithms. We present here a generalization
for score functions with different sensitivities.

Suppose the data set comes from a universe $U$ equipped with an
neighbor relation (e.g., Hamming or set-difference distance for
standard\anote{is this a good term?} data sets,
or vertex distance on graphs). We assume that the set of possible answers is finite and index it by elements of $[k]$.\snote{Added this sentence to clarify what $k$ is. Can we deal with infinite sets of answers?}
Given a collection of functions $q_1,...,q_k$ from  $U$ to $\re$ and a private data set
$x\in U$, the goal is to minimize $q_i(x)$, that is, to find an index $\hat
\i$ such that $q_{\hat \i}(x)\approx \min_i q_i(x)$.
Define
$$\Delta_i \defeq \max_{x,x'\in U\text{ adjacent}}
|q_i(x)-q_i(x')| \qquad \text{and}\qquad
 \Delta_{max} \defeq\max_i \Delta_i
.$$

The exponential mechanism achieves the following accuracy guarantee:
for every $\beta>0$,
with probability $1-\beta$, the output $\hat \i$ satisfies
$q_{\hat \i}(x)\leq \min_i q_i(x) + \Delta_{max}\cdot
\frac{2\log(k/\beta)}{\eps}$.\snote{Define its efficient
  implementation, ReportNoisyMean and give references.}\anote{Now in
  preliminaries of the full version.}

A limitation of this guarantee is that it depends on the maximum
sensitivity of the score functions $q_i(\cdot)$. In the context of
threshold selection for graph algorithms, such a guarantee is
meaningless for sparse graphs. This poor utility bound is not merely an artifact
of the analysis. The problem is inherent in the algorithm. For
example, consider the setting with $k=2$, where the two score functions have sensitivity
$\Delta_1=1$ and $\Delta_2\gg 1$. Further, consider a data set $x$ with $q_1(x)=0$ and
$q_2(x)=\Delta_2/\eps$. On input $x$, the exponential mechanism will select $\hat \i
=2$ with constant probability, resulting in an excess error of
$\Delta_2/\eps$, which may be arbitrarily larger than $\Delta_1$.

In contrast, we give an algorithm whose excess error scales
with
%depends on
\snote{Is proportional to?}\anote{better?} the
sensitivity of the \emph{optimal} score function $\Delta_{i^*}$, where
$i^*=\argmin_i q_i(x)$. Our mechanism requires as input an upper bound $\beta$ on the desired probability of
a bad outcome; the algorithm's error guarantee depends on this $\beta$.

\newtheorem*{detailedEMThm}{Theorem~\ref{thm:genEM}}
\begin{detailedEMThm}[Formal]
  For all parameters $\beta\in (0,1)$, $\eps>0$, the
  generalized exponential mechanism (Algorithm~\ref{alg:genEM}) is
  $(\eps,0)$-differentially private (with respect to the neighbor relation on $U$). For all inputs $x$, the output $\hat \i$ satisfies\snote{Consider changing notation for scores to $q_i(x)$.}\anote{done.}
  \begin{equation}
  q_{\hat \i}(x) \leq \min_i \paren{ q_i(x) +  \Delta_i \cdot
  \tfrac{4\log(k/\beta)}{\epsilon} }\, .\label{eq:genEM}
\end{equation}

\end{detailedEMThm}

In particular, our algorithm is competitive with the sensitivity of the true
minimizer $i^*=\argmin_iq_i(x)$ (since the right-hand side of
\eqref{eq:genEM} is at most $q_{i^*}(x)+\Delta_{i^*} \cdot
  \tfrac{4\log(k/\beta)}{\epsilon}$). In the case that all the
  $\Delta_i's$ are the same, our algorithm simplifies to running the usual
  exponential mechanism with $\epsilon'=\epsilon/2$; this justifies
  the ``generalized'' name.

\begin{algorithm}[tH]\label{alg:genEM}
  \caption{Generalized Exponential Mechanism}
  \KwIn{Data set $x$ from universe $U$, parameters $\beta \in
    (0,1)$ and $\eps>0$, \\
score
    functions $q_1,...,q_k$ from $U$ to $\re$.}
  Set $t = \frac {2\log(k/\beta)}{\eps}$\;
\For{$i=1$ to $k$}{
  $\Delta_i \defeq \max_{x,x'\in U\text{ adjacent}}
|q_i(x)-q_i(x')| $.  \tcc{An upper bound on $\Delta_i$ suffices. \\ Generally, $\Delta_i$
  is known exactly.}}
\For{$i=1$ to $k$}{
  $\displaystyle s(i) \gets \max_{j\in [k]} \frac{(q_i(x) + t \Delta_i)
    -(q_j(x) + t\Delta_j)}{\Delta_i +\Delta_j}$\tcc{$s(i)$ has
    sensitivity at most 1.}
}
  \Return $\hat \i \gets ExponentialMechanism(s(i), i\in [k],\epsilon)$\;
\end{algorithm}

The intuition behind the algorithm is as follows: since the score
function $q$ has different sensitivity for each $i$, we would like to
find an alternative score function which is less sensitive. One simple
score would be to compute, for each $j$, the \emph{distance}, in the space of data sets, from
the input $x$ to the
nearest data set $y$ in which $q_j(y)$ is smallest among the values
$\{q_i(y)\}_{i\in [k]}$ (this idea is inspired by the GWAS algorithms
of \cite{JohnsonS13}. This score has two major drawbacks: first,
it is hard to compute in general; second, more subtly, it will tend to
favor indices $j$ with very high sensitivity (since they become
optimal with relatively few changes to the data).

Instead, we use a substitute measure which is both easy to compute
(given the scores $q_i(x)$ for $i\in [k]$) and appropriately penalizes
scores with large sensitivity. Given a value $t>0$ (to be set later), define the \emph{normalized score} as
\begin{equation}
s(i;x) = \max_{j \in [k]} \frac{(q_i(x) + t \Delta_i)
    -(q_j(x) + t\Delta_j)}{\Delta_i +\Delta_j}
= \max_{j \in [k]} \paren{\frac{q_i(x)
    -q_j(x) }{\Delta_i +\Delta_j} + t \cdot \frac{\Delta_i - \Delta_j}{\Delta_i +\Delta_j} }
\, .\label{eq:s}
\end{equation}

The first term inside the maximum on the right-hand side is an
approximation to the Hamming distance from $x$ to the nearest data set
$y$ where score $q_j(\cdot)$ becomes smaller than $q_i(\cdot)$. The
second term (containing $t$ and independent of the data set) penalizes indices $i$ with larger
sensitivity.

We obtain an index $\hat \i$ by running the usual exponential mechanism
on the normalized scores $s(i)$.
Our first lemma bounds the sensitivity of the normalized score.

\begin{lemma}
  For each $i$, and for any $t\in \re$, the normalized score
  $s(i;\cdot)$ has
  sensitivity at most 1.
\end{lemma}

\begin{proof}
  First, fix indices $i,j\in [k]$. The ratio $\frac{q_i(x) -q_j(x)
  }{\Delta_i +\Delta_j}$ has sensitivity at most 1 since $q_i(\cdot)$
  and $q_j(\cdot)$ can vary by at most $\Delta_i$ and $\Delta_j$,
  respectively, when $x$ changes to an adjacent data set. As long as
  $t$ does not depend on $x$, the function $s(i;\cdot)$ is a maximum
  of sensitivity-1 functions, which means its sensitivity is at most
  1.
\end{proof}

\begin{proof}[Proof of Theorem~\ref{thm:genEM}]
  The algorithm sets $t = 2\ln(k/\beta)/\eps$, regardless of the data
  $x$. Since the normalized scores have sensitivity at most 1, the
  application of the usual exponential mechanism (or its more
  efficient alternative, ``report noisy min'') is
  $(\eps,0)$-differentially private.

  To analyze utility, let $\tilde \i$ denote the index that minimizes
  the penalized score $q_i(x)+ t \Delta_i$. Then
 $$s(\tilde
  \i;x) =0,$$
  since each of the terms in the maximum defining $s$ is nonpositive
  for $\tilde \i$ (and the term for $j=\tilde \i$ is 0). By the usual
  analysis of the exponential mechanism, we have that with probability
  at least $1-\beta$,
  $$s(\hat \i;x) \leq  \underbrace{s(\tilde \i;x)}_{0} + \frac{2
    \ln(k/\beta)}{\eps}\, .$$
  Now consider an arbitrary index $j$. Since $s(\hat \i;x)$ is at
  least $\frac{(q_{\hat \i}(x) + t \Delta_{\hat \i})
    -(q_j(x) + t\Delta_j)}{\Delta_{\hat\i} +\Delta_j}$, we can multiply by
  $\Delta_{\hat \i} +\Delta_j$ to obtain:
  \begin{eqnarray*}
    q_{\hat \i}(x)
    &\leq& q_j(x) + t\paren{\Delta_j-\Delta_{\hat\i}} +
    \tfrac{2 \ln(k/\beta)}{\eps}\cdot \paren{\Delta_{\hat\i} +
      \Delta_j} \\
    &=& q_j(x) +\Delta_{j} \paren{\tfrac{2 \ln(k/\beta) }{\eps} +
    t} + \Delta_{\hat\i}\paren{\tfrac{2 \ln(k/\beta) }{\eps}-t} \, .
  \end{eqnarray*}
  Substituting $t=\tfrac{2 \ln(k/\beta) }{\eps}$ yields the desired
  result.
\end{proof}

\ifnum\full =1 %%% MATCHES EM DETAILS
\subsection{Threshold Selection for Lipschitz Extensions}
\label{sec:threshold}

Suppose we have a collection of candidate functions\snote{Why not just
  $f_i$?}\anote{To match notation elsewhere, where $\thresh$ is the subscript
for the extension.} $\{f_{\Delta_i}\}_{i \in
  [k]}$ for approximation a function $f$, each with sensitivity $\Delta_i$. Moreover, the approximation functions
are all underestimates, that is,
$$f_{\Delta_i}(G) \leq f(G) \qquad \text{ for all } G\text{ and
}\Delta_i\, .$$
Let $$err(\Delta_i) = |f(G)-f_\Delta(G)| +
\Delta/\eps.$$  This is a simple proxy for the (expected) error in
approximating $f(G)$ that one gets by using $f_{\Delta}(G) +
\Lap(\Delta/\eps)$. It exaggerates the expected error by a factor of
at most 2, since the expected error is at most $err(\Delta)$ by the
triangle inequality, and at least $\min(|f(G)-f_\Delta(G)| ,
\Delta/\eps))$.

The functions $err(\Delta_i)$ don't necessarily have bounded
sensitivity (since we make no assumption on how $f$ varies).
However, the differences
$err(\Delta_i)-err(\Delta_j)$ do have sensitivity at most
$\Delta_i+\Delta_j$, which allows us to employ the generalized
exponential mechanism (alternatively, since the functions
$f_{\Delta_i}$ are all underestimates, we may use $q_i(x) =
-f_\Delta(G)+\Delta/\eps$).

\begin{corollary} Running the generalized exponential mechanism with
  score $q_i(x) = err(\Delta_i)$ and sensitivities $\Delta_i$ is
  differentially private and
  yields a threshold $\hat \Delta $ such that, for every
  $\Delta^* \in \{\Delta_i\}$, with probability at
  least $1-\beta$,
  $$err(\hat \Delta)\leq err(\Delta^*) + \frac{4\Delta^*
    \log(k/\beta)}{\eps}
  \leq err(\Delta^*) \cdot O\paren{\ln  \paren{\frac k  \beta }}
  \, .$$
\end{corollary}

\paragraph{Selecting from a continuous interval of thresholds}
The extensions we consider satisfy a further guarantee of
\emph{monotonicity}, namely, if $\Delta_1<\Delta_2$, we have
\begin{equation}
f_{\Delta_1}(G) \leq  f_{\Delta_2}(G) \leq f(G)\, .\label{eq:monodelta}
\end{equation}
If we want to select among an interval $[1,\Delta_{max}]$ of possible
thresholds, then this guarantee ensures that selecting among the
powers of a fixed constant (e.g., $1, 2, 4, ..., 2^{\lfloor
  \log_2(\Delta_{max})\rfloor}$) will still give a multiplicative
approximation to be best choice of $\Delta$, since for all values of $\Delta\geq 0$,
$$err(\Delta) \leq 2 err(\Delta/2).$$

We obtain the following proposition.

\begin{prop}\label{prop:selectGen}
  If the collection of functions $\{f_{\Delta}\}_{\Delta \in
    [1,\Delta_{max}]}$ forms a monotone family of approximations to
  $f$ (as in \eqref{eq:monodelta}), then applying the generalized
  exponential mechanism to the powers of $2$ in the interval $[1,\Delta_{max}]$ yields a threshold $\hat \Delta$ such that, for every
  $\Delta^*\in [1,\Delta_{max}]$, with probability at least $1-\beta$,
  $$err(\hat \Delta) = err(\Delta^*) \cdot
  O\paren{\ln\ln(\Delta_{max})  + \ln \frac 1\beta }
  \, .$$
%  where  $\displaystyle \Delta^* = \min_{\Delta\in [1,\Delta_{max}]} err(\Delta)$.
\end{prop}

This generalizes and improves on the result of \citet{ChenZ13}, who gave
a method for selecting a sensitivity threshold that was specific
to their Lipschitz extensions and  within a $\log(n)$
multiplicative factor (as opposed to $\log \log n$) of the optimal error.

\subsection{Selecting a Threshold for the Degree Distribution}

Consider the algorithm for releasing the degree distribution discussed
in Section~\ref{sec:degree-list-extension}. Recall that the
algorithm's $\ell_1$ error is at
most
\begin{equation}
err(\thresh) \defeq \|\dlext(G) -\deglist(G)\|_1 + 6\thresh^2/\eps\, .\label{eq:histerror}
\end{equation}

This error function is closely related to the error of the
approximation to the number of edges from
\citet{KNRS13}. Specifically, let $g(G)$ denote the number of edges in
$G$, and $g_\thresh$ denote the Lipschitz extension of $g$ from $\Gd$
to $\G$. Then $$g_\thresh(G) = \|\dlext(G)\|_1
\qquad
\text{and}\qquad
g(G) -
g_\thresh(G) = \|\dlext(G) -\deglist(G)\|_1 \, .$$

We can therefore use the process above for selecting a threshold for a
one-dimensional function to select a threshold for releasing the
degree distribution.

\begin{prop} \label{prop:selectDH}
  Given a graph $G$ and parameter $\eps$, let $\thresh^* = \argmin_{\thresh \in [1,n]} err(\thresh)$.
  Applying the generalized
  exponential mechanism with $q_i(G)= err(2^i)$, for $i\in
  \{1,2,4,...,2^{\lfloor \log_2(n)\rfloor}\}$ is
  $(\eps,0)$-differentially private and yields a threshold $\hat
  \thresh$ such that, for every $\Delta^*\in [1,\Delta_{max}]$, with
  probability at least $1-\beta$,
  $$E \paren{err(\hat \thresh)} \leq
  2err(\thresh^*) + \frac{8\thresh^* \ln (\ln (n) /\beta)}{\eps}
=
 err(\thresh^*) \cdot O\paren{\ln\ln(n) + \ln \frac 1\beta }
\, .$$
\end{prop}

\section{Error Analysis on $\alpha$-Decaying Graphs}
\label{sec:decayerror}

Our techniques provide a significantly more accurate way to release
the degree distributions of graphs while satisfying node-level
differential privacy. To illustrate this, we study the accuracy of our
method on graphs that satisfy $\alpha$-decay, a mild condition on the
tail of the degree distribution.

\subsection{$\alpha$-Decay}\label{sec:assumption}
Recall that  $\Av(G) = 2|E|/|V|$ is the average degree of $G$.

\begin{assumption}[$\alpha$-decay]\label{assum:decay}
Fix $\alpha\geq 1$.  A graph $G$ satisfies {\em $\alpha$-decay} if for all\footnote{Our results hold even when this condition is satisfied only for sufficiently large $t$\drop{namely, at least polynomial in the number of nodes}. For simplicity, we use a stronger assumption in our presentation.} real numbers  $t>1$, $\cum_G(t\cdot \Av)  \leq  t^{-\alpha}$.
\end{assumption}

Note that \emph{all} graphs satisfy 1-decay (by Markov's
inequality). The assumption is nontrivial for $\alpha >1$, but it is
nevertheless satisfied by almost all widely studied classes of
graphs. So-called ``scale-free'' networks (those that exhibit a
heavy-tailed degree distribution) typically satisfy $\alpha$-decay for
$\alpha \in (1,2)$.\anote{citation?} Random graphs satisfy
$\alpha$-decay for essentially arbitrarily large $\alpha$ since their
degree distributions have tails that decay exponentially (more
precisely, for any $\alpha$ we can find a constant $c_\alpha$ such
that, with high probability, $\alpha$-decay holds when $t>c_\alpha$).
Regular graphs satisfy the assumption with $\alpha=\infty$. The
following lemma bounds the number of \emph{edges} adjacent to nodes
with degree above a given threshold.

\begin{lemma}\label{lem:implication-of-decay}
Consider a graph $G$ on $n$ nodes that satisfies $\alpha$-decay for
$\alpha>1$, and let $\thresh>\Av(G)$. Then
$$\sum_{v:\ \deg_v(G)> \thresh} \leq
\frac{\Av(G)^\alpha}{(\alpha+1)\thresh^{\alpha-1}} \cdot n\, .$$
\end{lemma}

\subsection{Error Analysis}
\label{sec:alpha-analysis}

\citet{KNRS13} gave algorithms for releasing the degree distribution
using a projection-based technique. Their algorithm required knowledge
of the decay parameter $\alpha$ (which was used to select the
projection threshold). They bounded the $\ell_1$ error of their
algorithm in estimating the degree distribution, and showed that it
went to 0 as long as $\alpha>2$ and $\Av$ was polylogarithmic in
$n$. More precisely, they gave an expected error bound of
$$\E \|\hat p - p_G \|_1 = \tilde O\paren{ {\Av^\frac{3\alpha}{\alpha+1} }/\paren{{\eps^2 n^{\frac{\alpha-2}{\alpha+1}}}}}\, .$$

Combining the noisy Lipschitz extension of the degree histogram
(Theorem~\ref{thm:ell1errorfixedthreshold}) with the threshold
selection algorithm (Proposition~\ref{prop:selectDH}), we get an
algorithm $\A_{combo}$ with much better accuracy guarantees that,
additionally, does not need to know the parameter $\alpha$.

\begin{algorithm}[ht]\label{alg:combo}\caption{Degree Histogram
    Estimation For Unknown Threshold}
$\hat D \gets$ Generalized Exponential
Mechanism$(q_\thresh(\cdot), \thresh \in\{1,2,4,....,2^{\lfloor \log n \rfloor}\})$ using score
$q_\thresh( G) = err(\thresh)$ and sensitivity bound $\thresh$ \;

$\hat h \gets \dlexthat(G) + (Y_1,\ldots,Y_{\hat \thresh})$
where $Y_i\sim \Lap(4\hat \thresh / \epsilon)$ i.i.d. \;

\Return $\hat p = \hat h / \|\hat h\|_1$\;

\end{algorithm}

%\anote{Write the whole algorithm (with normalization)}

\begin{theorem}
  Given inputs $G\in \G$ and $\eps>0$, the algorithm $\A_{combo}$
  produces an estimate $\hat p$ such that, if $G$ satisfies $\alpha$
  decay for $\alpha>1$, then
$$\E\|\hat p - p_G \|_1 = O\paren{ {\Av^\frac{2\alpha}{\alpha+1} (\ln
    \ln n)
  }/\paren{\eps n}^{\frac{\alpha-1}{\alpha+1}}} \, .
$$
\end{theorem}

In particular, this error is $o(1)$ as $n$ goes to infinity if
$\alpha>1$ and  $\Av^{\frac{2\alpha}{\alpha-1}}=o(\eps n)$.

\begin{proof}
  Fix a graph $G$ that satisfies $\alpha$ decay for $\alpha>1$, and
  let $\Av$ denote its average degree.  Conditioned on selecting a
  given degree threshold $\thresh$,
  Theorem~\ref{thm:ell1errorfixedthreshold} guarantees that the
  $\ell_1$ error of our algorithm in estimating $hist(\deglist(G))=n
  \cdot p_G$ is at most $err(\thresh)= 2 n\sum_{i> D} P_G(i) +
  \frac{6\thresh^2 }{\epsilon}$ (defined as in \eqref{eq:histerror}).

  Although the true size of the graph $n$ is not known to the
  algorithm it is convenient to divide everything by $n$ so that we
  can compare to the true degree distribution $p_G$. Let $\tilde p =
  \hat h / n$ denote
  the estimate of $p_G$ one gets by normalizing the estimated degree
  histogram by the true vertex count $n$
  rather than $\|\hat h\|_1$. We will account
  for the estimation of $n$ at the end of the proof.  Dividing by $n$,
  we get a bound of the form
  $\frac{err(\thresh)}{n}=O\paren{\sum_{i> D} P_G(i) +
    \frac{\thresh^2 }{\eps n}}$ on the error in estimating $p_G$. By
  Lemma~\ref{lem:implication-of-decay}, this bound is at most
  $O\paren{\Av^\alpha /((\alpha+1)\thresh^{\alpha-1}) +
  \thresh^2/(n\epsilon)}$.
  In particular, if we set $\thresh^* = \paren{\Av^\alpha \eps n
  }^{1/(\alpha+1)}$ (which makes the two terms in the sum equal) then
  the expected $\ell_1$ error of $\tilde p$ is at most
$$\frac{err(\thresh^*)}{n}=O(  {\Av^\frac{2\alpha}{\alpha+1}
  }/\paren{\eps n}^{\frac{\alpha-1}{\alpha+1}}).$$

  In the algorithm, we do not select $\thresh^*$ but rather a
  differentially private alternative $\hat \thresh$.  By the law of
  conditional expectations, the overall expected error of $\tilde p$
  is at most the expectation of $\frac{err(\hat \thresh)}{n}$, that is,
  $$\textstyle \E \|\tilde p  - p_G\|_1  = \E_{\hat\thresh}\paren{
  \E_{noise} \paren{\|\tilde p - p_G\|_1 \ \Big| \ \hat \thresh}}
\leq \frac 1 n \E_{\hat\thresh} \paren{err(\hat \thresh)}\,.$$
By Proposition~\ref{prop:selectDH}, the expectation of
  $err(\hat\thresh)$ is at most $2err(\thresh^*) + 8 \thresh^* \ln \ln
  (n)/\eps$. For $\alpha<\infty$, the reference
  threshold $\thresh^*$ is polynomially large in $n$. Thus, the first
  term
  $err(\thresh^*)$ (which is at least $(\thresh^*)^2/\eps$) dominates
  the second term, and the final error bound is
  $$\textstyle \E \|\tilde p- p_G\|_1  = O(  {\Av^\frac{2\alpha}{\alpha+1}
  }/\paren{\eps n}^{\frac{\alpha-1}{\alpha+1}})\, .$$

  Finally, we analyze the difference between $\A_{combo}$ and
  $\A'$. Let $\hat n$ denote the estimated number of edges in $G$,
  that is $\hat n = \|\hat h\|_1$. Note that  for any given
  realization of the algorithm's random choices, if $\tilde p$ is a
  good approximation to the true distribution $p_G$, then $\hat n$
  must be good estimate of the true number of vertices:
  $$|\hat n - n|  = |\ \|\hat h
  \|_1 - n\ | \leq n |\ \|\tilde p\|_1 -1\ | \leq n \|\tilde p -  p_G \|_1
  \, .$$
  This allows us to bound the difference between $\hat
  p$ and $\tilde p$. Since
$
\tilde p - \hat p = \hat p \paren{\frac{\hat n}{n} -1 } $, the
$\ell_1$ norm of $\tilde p - \hat p$ is at most $\|hat p\|_1 \|\tilde
p -  p_G \|_1 = \|\tilde p -  p_G \|_1$. Thus, the error of $\hat p$
in estimating $p_G$ is never more than twice the error of $\tilde p$:
\begin{displaymath}
\| \hat p -p_G\|_1 \leq   2\| \tilde p -p_G\|_1
\qquad
\text{ and thus }
\qquad
 \textstyle \E \|\hat p- p_G\|_1  = O(  {\Av^\frac{2\alpha}{\alpha+1}
  }/\paren{\eps n}^{\frac{\alpha-1}{\alpha+1}})\, . \qedhere
\end{displaymath}
\end{proof}

\fi %%% MATCHES EM DETAILS

\begin{small}
\bibliographystyle{abbrvnat}
\bibliography{Lipschitz-extensions}
\end{small}

\ifnum\full =1 %%% DOUBLING

\appendix
\section{Doubling dimension of graph metrics}
\label{sec:doubling}

\begin{definition}
  The doubling dimension $\rho$ of a metric space $(X,d)$ is the
  smallest integer such that every ball in $X$ can be covered using at
  most $2^\rho$ balls of half the radius.
\end{definition}

Consider the set $\Gn$ of graphs on at most $n$ vertices. If we equip
$\Gn$ with the edge adjacency metric, we get a set essentially
equivalent to the $\binom n 2$-dimensional Hamming cube (in fact, a
union of $n$ different Hamming cubes corresponding to graphs of sizes
$1,2,...,n$).  This metric has doubling dimension $\Theta(n^2)$.

Intuitively, the doubling dimension of the vertex-adjacency metric on
$\Gnd$ should be similar. We sketch a weaker statement here, namely
that the doubling dimension is $\Omega(n)$. This bound shows that
constructions with stretch bounded by the doubling dimension still
have very high stretch when used on the vertex metric.

\begin{lemma}
  The doubling dimension of the vertex-adjacency metric on $\Gnd$ for
  $\thresh\geq 1$ is $\Omega(n)$. If we collapse the set $\Gnd$ by
  identifying isomorphic graphs, then the statement continues to hold
  for $\thresh\geq 4$.
\end{lemma}

\begin{proof}[Proof Sketch]
  Assume $n$ is even, w.l.o.g.\ To prove the theorem, we embed the
  Hamming cube $Ham_{n/2}$ into $\Gnd$, which shows that $\Gnd$ has
  doubling dimension $\Omega(n)$. Let $G_0$ be a uniformly random
  regular graph of degree 3 on $n/2$ vertices. For every subset
  $S\subseteq [\frac n 2]$, let $G_S$ be the graph on $\frac n 2 +|S|$
  vertices obtained by starting from $G_0$ and adding one vertex for
  each element in $S$ and connecting it to the corresponding vertex in
  $G_0$.

  The vertex distance between two such graph $G_S$ and $G_T$ is
  $\Omega(|S \triangle T|)$, as with the Hamming metric, as long as
  $S$ and $T$ are sufficiently far from each other. This is sufficient
  to prove the main result as we may select $S$ and $T$ from an
  error-correcting code with linear rate and minimum distance.

  A
  complete proof is delicate, since one must account for the
  possibility that one can get from $G_S$ to a graph that is
  isomorphic to a subgraph of  $G_T$ by deleting fewer than $|S\triangle T|$ vertices
  from $G_0$. We omit the details.
\end{proof}

\fi %%% DOUBLING

\ifnum\extras=1 %MATCHES EXTRA APPENDICES

\section{$\ell_2$ Lipschitz extension of the adjacency matrix}
\label{sec:adjacencyext}

In this section, we give a Lipschitz extension of the adjacency matrix
of the graph in the $\ell_2$ norm -- in other words, an algorithm that
``projects'' to a fractional graph which changes smoothly as a
function of the input.

\anote{Does this work directly for the Laplacian, too? Might be nice
  for revealing statistics about the graph spectrum?}

Given a (possibly directed) $G=(V,E)$ with adjacency matrix $A$,
a fractional subgraph of $G$ is an assignments of weights in $[0,1]$
to the edges of $G$. We represent it as a matrix $X$ which has an
entry $x_{uv} \in [0,1]$ for each 1 in $A$, and 0's elsewhere.
When $G$ is undirected, we restrict $X$ to be symmetric.

The fractional in-degree of a vertex $u$ in $X$ is  $d_v^{in}(X) =
\sum_{u \in V} x_{uv}$; out-degree is defined similarly.
Given $G$, we let $C_{G,\thresh}$ of set of fractional subgraphs with in- and
out-degree bounded above by $\thresh$, that is
$$C_{G,\thresh} = \left\{ X \in \re^{V \times V} \ : \
\begin{matrix}
  0\leq X_{i,j} \leq A_{i,j} &\text { for all }i,j \, , \\
  \sum_j
  X_{i,j}\leq \thresh & \text{ for all }i \, , \\
\sum_i
  X_{i,j}\leq \thresh &\text{ for all }j \\
\end{matrix}
\right\} \, .$$

% Given a graph $G$, let $F^{in}_{G,\thresh}$ denote the polytope of possible
% in-degree lists, and $F^{out}_{G,\thresh}$ the possible out-degree
% lists of fractional graphs in $C_{G,\thresh}$ (for symmetric graphs, we drop the superscripts). Note
% that these sets are different from the sets of vectors in $[0,\thresh]$ that are
% dominated by the in- and out-degree lists of the graph $G$.

We will denote by $\re^{*\times *}$ the set of real-valued square
matrices of finite dimension. We assume that for each such matrix $X$,
the rows and columns are indexed by a set $V$, and we compare two
matrices $X$ and $X'$ indexed by sets $V,V'$ by extending each to be
indexed by $V\cup V'$, filling in new entries with 0. (In the language
of graphs, this corresponds to padding the two graphs with isolated
vertices to make the vertex sets identical). We use $\re^{*\times
    *}_{\ell_2}$ to denote the metric space with the  $\ell_2$
metric.

Consider the function that maps a graph $G$ to its adjacency
matrix. Restricted to directed graphs with maximum in- and out-degree
at most $\thresh$, this function has $\ell_1$ sensitivity $2\thresh$,
and $\ell_2$ sensitivity $\sqrt{2\thresh}$.

Given a graph $G=(V,E)$, let $\hat A_\thresh(G)$ be given by the $\ell_2$
projection of the $|V|\times |V|$ all-ones matrix, denoted $\onesmat{|V|}$, onto
the polytope $C_{G,\thresh}$, that is:
$$\hat A_\thresh(G) \defeq \argmin_{X \in C_{G,\thresh}} \| \onesmat{|V|}- X\|_2^2
\,.$$
The projection is uniquely defined since $C_{G,\thresh}$ is convex and closed.

\begin{theorem} For every $\thresh \geq 0$,
  $\hat A_\thresh$ is an extension of $A:\Gd \to \re^{*\times
    *}_{\ell_2}$ to all of $\G$ with Lipschitz
  constant at most $4\thresh$, that is, with stretch at most 2.
\end{theorem}

\begin{proof} Let $G,G'$ be vertex-adjacent graphs, where $G$ is obtained
  by removing a node $\vnew$ along with its associated edges from a
  graph $G_2$. Let $n$ denote the size of the smaller graph
  $G$. Order the vertices of the two graphs so that the first $n$
  vertices are the same for both graphs, and $\vnew$ is the $n+1$-st
  vertex in $G'$.

  Let $x^*_{old}=\hat A_\thresh(G)$ and $x^*_{new}=\hat
  A_\thresh(G')$. These matrices have different dimensions ($n$ and
  $n+1$), but by adding a row/column of zeros to $n$-dimensional
  matrices, we may view $C_{G,\thresh}$ as a subset of
  $C_{G',\thresh}$. With this view, there is no difference between
  minimizing the distance to $\onesmat{n}$ and the distance to
  $\onesmat{n+1}$. Thus, $x^*_{old},x^*_{new}$ are minimizers of the same
  function $\|\onesmat{n+1}-\cdot\  \|_2^2$ over the sets
  $C_{G,\thresh}\subseteq C_{G',\thresh}$, respectively. To prove the
  theorem, we must show that $\|x^*_{new} - x^*_{old}\|_2^2\leq 4\thresh$.

  Write $x^*_{new}$ in block form $x_0 + x_1$,
  where $(x_0)_{i,j}= (x^*_{new})_{i,j}$ for $i,j \in [n]$ and $(x_0)_{i,j}= 0$
  elsewhere, and $(x_1)_{i,j}= (x^*_{new})_{i,j}$ for pairs $i,j$
  where either $i$ or $j$ is $n+1$. Note that $x_0$ and $x_1$ have
  disjoint support, and the support of $x_1$ is disjoint from that of
  $x^*_{old}$.
Hence,
  $$\|x^*_{new} - x^*_{old}\|_2^2 = \|x_1\|_2^2+ \|x_0 - x^*_{old}\|_2^2 \,.$$

  We start by bounding the second term. Note that $x_0$ is in the set $C_{G,\thresh}$ (the feasible set for the
  smaller graph $G$), since
  $x_0$ has row- and column-sums bounded by $\thresh$ and the
  last row and column of $x_0$ are all 0's.

%
  % Since the function $\|\onesmat{n+1} - \cdot\
  % \|_2^2$ is 2-strongly convex with respect to the $\ell_2$ norm, and
  % since $x^*_{old}$ is its minimizer on the set $C_{G,\thresh}$,
  Since $x^*_{old}$ is the minimizer of $\|\onesmat{n+1} - X\
  \|_2$ on the set $C_{G,\thresh}$, the angle formed by the segments $\onesmat{n+1}
  \to x^*_{old} \to X$ is at least $90^\circ$, for any point
  $X\in C_{G,\thresh}$. (This holds because  $C_{G,\thresh}$ is
  convex. If the angle were acute, one could move along the segment
  from $x^*_{old} $ to $X$ and get closer to $\onesmat{n+1}$).
Thus,
  $$\| \onesmat{n+1} - X \|_2^2 \geq \| \onesmat{n+1} - x^*_{old}\|_2^2 + \|  x^*_{old}-X\|_2^2\,.$$
  Rearranging terms and substituting $X =  x_0$, we get
  \begin{equation}
 \| x_0-  x^*_{old}\|_2^2 \leq
\| \onesmat{n+1} - x_0\|_2^2 - \| \onesmat{n+1} - x^*_{old}\|_2^2\,.\label{eq:3}
\end{equation}

We can develop the distance from $\onesmat{n+1}$ to $x^*_{old}$ in terms of the distance from
$\onesmat{n+1}$ to $x^*_{new}$:
$$
\begin{matrix}
  \| \onesmat{n+1}- x_0\|_2^2
  &=& \| \onesmat{n+1} - x_0-x_1 \|_2^2
  &+&2 \ip{ (\onesmat{n+1}-x_0) , x_1 } &-& \|x_1\|_2^2  \ \\
  &\leq& \| \onesmat{n+1}- x^*_{new}\|_2^2 &+& 2\|x_1\|_1 &-&
  \|x_1\|_2^2 \,,
\end{matrix}
$$
where we have used the fact that the entries of $\onesmat{n+1}- x_0$
are in $[0,1]$ to bound the inner product with $x_1$.
Now, recall that both $x^*_{new}$ and $x^*_{old}$ can be viewed as
elements of the set $C_{G',\thresh}$ (the feasible set for $G'$), and $x^*_{new}$
minimizes the distance to $\onesmat{n+1}$ within that set. Plugging the
bound above into \eqref{eq:3}, we get
$$\| x_0-  x^*_{old}\|_2^2 \leq \underbrace{\| \onesmat{n+1}-
  x^*_{new}\|_2^2  -  \| \onesmat{n+1} - x^*_{old}\|_2^2}_{\leq 0} +
2\|x_1\|_1  - \|x_1\|_2^2\, .$$
Putting the pieces together, we have
  $$\|x^*_{new} - x^*_{old}\|_2^2 = \|x_1\|_2^2+ \|x_0 -
  x^*_{old}\|_2^2 \leq 2\|x_1\|_1 \,.
$$
Finally, the $\ell_1$ norm of $x_1$ is at most $2\thresh$ since the
  entries of row $n+1$ of $x^*_{new}$ add to at most $\thresh$, as do the
  entries of column $n+1$. We get $\|x^*_{new} - x^*_{old}\|_2^2 \leq
  4\thresh$, as desired.
\end{proof}

\section{Other Applications}
\label{sec:otherapps}

\begin{itemize}
\item
  edge types?

\item   $(in,out)$ degree pairs? (See Sofya's ``new'' Google notebook for
  10/09/14 and 10/15/14).

\item Analyze our algorithms on specific stochastic families, e.g.,
  Sesa paper on $\beta$ model? Product graphons?

\item Use $\ell_2$ extension for graphon approximation.
\end{itemize}

\subsection{Generalizations to Related Functions}
\label{sec:bluepink}

The ideas in the preceding two sections can be generalized to a variety
of settings. For example, suppose that vertices have

Suppose that edges are divided in $k$ types, so that the degree
distribution real y

\fi %MATCHES EXTRA APPENDICES

\end{document}